\documentclass[letterpaper,10pt,twocolumn]{IEEEtran}

\usepackage{amsmath,amssymb}
\usepackage{graphicx}
\usepackage{subfigure}
\usepackage[noadjust]{cite}
\usepackage{bm}
\usepackage[usenames,dvipsnames]{color}
\usepackage{algorithm}
\usepackage{algorithmic}
\usepackage{setspace}
\newtheorem{thm}{Theorem}

\floatname{algorithm}{Protocol}
\algsetup{indent=2em}

\newcommand{\KH}[3][]{{\color{black}#2}}%
\newcommand{\KHH}[3][]{{\color{black}#2}}%
\newcommand{\K}[3][]{{\color{black}#2}}%
\newcommand{\kh}[3][]{{\color{black}#2}}%
\newcommand{\khh}[3][]{{\color{black}#2}}%

\begin{document}

\title{Medium Access Control for Wireless Networks with Peer-to-Peer State Exchange\thanks{\K{K.~H.~Hui, D.~Guo and R.~A.~Berry}{Ka Hung Hui, Dongning Guo and Randall A.~Berry} are with the Department of Electrical Engineering and Computer Science, Northwestern University, Evanston, IL 60208 (Email: khhui@u.northwestern.edu, dguo@northwestern.edu, rberry@ece.northwestern.edu). The material in this paper was presented in part at the IEEE International Symposium on Information Theory, Saint Petersburg, Russia, July-August 2011. This work was supported by DARPA under grant W911NF-07-1-0028.}}

\IEEEoverridecommandlockouts

\author{Ka Hung Hui, Dongning Guo and Randall A.~Berry}
\date{\today}
\maketitle

\begin{abstract}
Distributed medium access control (MAC) protocols are proposed for wireless networks \KH{assuming that}{in which} one-hop peers can \khh{periodically}{} exchange a small amount of state information\khh{}{ periodically}. 
Each station maintains a state and makes state transitions and transmission decisions based on its state and recent state information collected from its one-hop peers.
\KH{A station can adapt its packet length and the size of its state space to the amount of traffic in its neighborhood}{Different stations can use different number of states and packet lengths}.
It is shown that these protocols converge to a steady state, where stations take turn\KH{s}{} to transmit in each neighborhood without collision. 
\KH{In other words}{An important consequence of this work is that using such protocols}, an efficient time-division multiple access (TDMA) like schedule \KH{is}{can be} formed in a distributed manner, as long as the topology of the network remains static or changes slowly with respect to the execution of the protocol.
\end{abstract}

\section{Introduction}
\label{sec:introduction}

\K{In most wireless networks, medium access control (MAC) is needed to avoid excessive collisions, which occur if a station transmits to another transmitting station\khh{,}{} which is half-duplex, or a station receives multiple simultaneous transmissions and cannot successfully decode \kh{the}{} desired message(s).}{The performance of wireless networks is \KH{limited}{degraded} by two types of collisions: a station transmits to another transmitting station which is half-duplex; or\KH{}{,} a station receives simultaneous transmissions from multiple stations such that all transmissions cannot be decoded correctly.
Medium access control (MAC) protocols are needed to eliminate or reduce both types of collisions. }
\khh{
Many MAC protocols can be viewed as requiring each station to maintain a state, which determines when the station transmits. 
To provide distributed operation, this state is updated based on information available locally in space.
For example, in carrier-sense multiple access (CSMA) protocols, the state is determined by the carrier sensing operation and the random backoff mechanism. 
}{
A typical design of MAC protocols is to let \KH{each station}{stations} maintain a state variable and make transmission decisions according to the state, which is updated based on the station's observations.
\KH{In general, \KHH{it is desirable to let}{} neighboring stations \KHH{use different states}{avoid using the same state at the same time} to avoid collisions. A}{Besides eliminating collisions, a} practical MAC protocol should \K{adapt to changes over time and}{} be \textit{distributed}, \textit{i.e.}, stations make their decisions based on information available locally in space. 
\K{}{A practical MAC protocol should also be \textit{adaptive} to recent changes in time.}
\KH{}{It is desirable to have a MAC protocol that offers high throughput and uses as few states as possible. }
}

\khh{This paper considers MAC protocols in which stations explicitly exchange limited state information. 
The protocols are \textit{self-stabilizing}, \textit{i.e.}, they converge to a collision-free schedule regardless of the initial state.}{\K{This paper studies}{We consider static or slow-varying networks and study} \textit{self-stabilizing} MAC protocols, \K{which are}{\textit{i.e.},} protocols that converge to a collision-free schedule, \K{regardless}{independent} of the initial state.}
\K{The underlying network is assumed to be static or vary slowly with respect to the execution of the MAC protocol.}{}
In the steady state, these protocols behave like time-division multiple access (TDMA), in which stations take turn\KH{s}{} to transmit without collision; while in the transient state, they behave like \khh{CSMA}{carrier-sense multiple access (CSMA)}, such that stations contend with each other, trying to find a slot for transmission and avoid collisions. 
\KH{Under the assumption of a single collision domain, \textit{i.e.}, all stations can hear each other,}{}\K{}{\KH{r}{R}eferences \cite{JLJW:1,JBBBCCMO:1,MFDMKDDL:1} study} self-stabilizing MAC protocols \K{have been studied in \cite{JLJW:1,JBBBCCMO:1,MFDMKDDL:1}}{\KH{}{ for a single collision domain, in which all stations can hear each other}}.
By learning transmission decisions of others, stations are able to find a collision-free schedule in a decentralized manner. 
\kh{As is pointed out in \cite{JLJW:1}, \khh{these}{such}}{\KH{Such}{However, the analysis applies only to a single collision domain.
It is pointed out in \cite{JLJW:1} that the}} protocol\KH{s}{ therein} cannot guarantee the formation of a collision-free schedule in \KH{case of}{} multiple collision domains \khh{and focus on schedules for unicast traffic}{\kh{}{ \KH{as is pointed out in \cite{JLJW:1}}{}}}.

\khh{
This paper focuses on establishing collision-free schedules for broadcast and multicast traffic in networks with multiple collision domains. 
It is well-known that multiple collision domains complicate scheduling due in part to hidden terminals and exposed terminals. 
For unicast traffic, state exchange in the form of RTS/CTS signaling can help alleviate these complications. 
However, this is not suitable when a station wants to broadcast a packet to all nearby stations. 
To facilitate this, we consider a richer form of state exchange. 
}{}

\khh{We build on work in \cite{KHDGRBMH:1} and \cite{KHDGRB:1}, which introduces self-stabilizing MAC protocols for one- and two-dimensional regular networks on lattices.}
{
\K{Many practical}{In practice, most} networks have multiple collision domains\K{,}{ and} hence the hidden terminal and exposed terminal problems \K{arise}{}.
\KH{D}{For one- and two-dimensional \textit{regular} networks with links between nearest neighbors only, d}istributed, adaptive, self-stabilizing MAC protocols \KH{for \K{the special case of}{} one- and two-dimensional \KHH{regular}{} networks on lattices with links between nearest neighbors\kh{}{ only} have been}{are} introduced in \cite{KHDGRBMH:1} and \cite{KHDGRB:1}, respectively. 
}
\KH{The technique is to divide time into periodic cycles, where each cycle is divided into slots.
A station maintains a \K{single}{} state and transmits only over the slot corresponding to its state.
Once the protocols converge, a periodic state pattern \K{(with immediate neighbors \khh{assuming}{assume} different states)}{} is formed throughout the regular network, and the maximum broadcast throughput is achieved.
If one directly applies these ideas to networks with arbitrary topologies, sufficiently many states are needed for stations with many neighbors, but \K{in a neighborhood with few stations}{for stations with few neighbors}, the wireless channel is underutilized \kh{because few states are occupied}{}.}{
\KH{The protocols in \cite{KHDGRBMH:1} and \cite{KHDGRB:1} do not apply}{It is nontrivial to extend \cite{KHDGRBMH:1,KHDGRB:1}} to networks with arbitrary topolog\KH{ies}{y}. 
\KH{Regardless of the network topology, it is usually}{Traditionally, it is} assumed that every station uses the same \KH{set of states and the same slot interval}{number of states, \textit{i.e.}, the same number of slots in each cycle. 
This implies that all slots have the same length}. 
A problem with this approach is that, we need sufficiently many states for stations with many neighbors, but \KH{in an area of few stations}{for stations with few neighbors}, the wireless channel will be underutilized.
Therefore, we introduce the notion of \textit{multiple resolutions}, \textit{i.e.}, a station having more neighbors uses a fine resolution (more states are used, each state corresponding to a shorter slot); while a station with fewer neighbors uses a coarse resolution (fewer states are used, each state corresponding to a longer slot). }

\khh{
There has been a significant work on MAC scheduling for networks that builds on the seminal max-weight algorithm \cite{LTAE:1}, and attempts to derive distributed, low complexity algorithms which approach the throughput-optimal performance of \cite{LTAE:1}. 
Examples include \cite{PCKKXLSS:1,XLNS:1,GSNSRM:1,EMDSGZ:1,AEAOEM:1,XLSR:1}. 
These approaches seek to adapt the resulting schedule to queue variations. 
Here, we instead consider a model with saturated traffic and seek to find fixed rate-based schedules, as in \cite{YYGDSS:1}. 
Such a schedule is naturally more useful for traffic that has a fixed long-term arrival rate. 
More bursty traffic can be accomodated by reserving some fraction of time for contention-based access as in \cite{IRAWMAJMMS:1}. 
}{}

\KH{
The main contributions of this paper are as follows:
\begin{enumerate}
\item
In Section~\ref{sec:model}, we introduce the concept of \textit{multiple resolutions}, \textit{i.e.}, a station having more neighbors uses a fine resolution (more states \KHH{in its state space}{are used}, each state corresponding to a shorter slot); while a station with fewer neighbors uses a coarse resolution (fewer states \KHH{in its state space}{are used}, each state corresponding to a longer slot). 
\item
In Sections~\ref{sec:1d} and~\ref{sec:2d}, multi-resolution MAC protocols are proposed for broadcast in one- and two- dimensional networks with arbitrary topologies, respectively.
These protocols guarantee every station a chance to transmit in each cycle. 
In addition, they achieve approximate proportional fairness in the sense that \K{a station's \kh{throughput}{rate} is approximately inverse\khh{ly}{} proportional to the node density in its neighborhood}{two stations with similar number of neighbors have similar rates}.
We show that in one-dimensional networks, stations can determine their resolutions in a distributed manner. 
The same also holds for two-dimensional networks under a mild condition. 
In case the condition is not met, we propose a mechanism for stations to dynamically change their resolutions until collisions do not occur in the entire network. 
\item
We show that the multi-resolution protocols can be applied to a more general setting. 
In Section~\ref{sec:mc}, we consider multicast traffic. 
In Section~\ref{sec:multich}, broadcast and multicast in networks with multiple orthogonal channels are considered. 
\end{enumerate}
}{
In this paper, we propose multi-resolution MAC protocols for one- and two-dimensional wireless networks with arbitrary topolog\KH{ies}{y}\KH{, \textit{i.e.}, a station having more neighbors uses a fine resolution (more states are used, each state corresponding to a shorter slot); while a station with fewer neighbors uses a coarse resolution (fewer states are used, each state corresponding to a longer slot)}{}. 
These protocols guarantee every station a chance to transmit in each cycle. 
In addition, they achieve \KH{approximate proportional}{} fairness in the sense that \KH{the rate of a station is inverse proportional to the node density in its neighborhood}{two stations with similar neighborhoods have similar rates}.
For one-dimensional networks, stations can determine their resolutions in a distributed manner. 
The same also holds for two-dimensional networks under a mild condition. 
In case the condition is not met, we propose a mechanism for stations to dynamically change their resolutions until collisions do not occur in the entire network. }
In all cases, the convergence of such protocols to a collision-free schedule is rigorously established. 
\KHH{To achieve the global optimum in terms of throughput is an NP-complete problem \cite{AETT:2,RRKP:1}, which is out of scope of this paper}{We do not consider maximizing throughput or minimizing the number of states; these problems are NP-complete \cite{AETT:2,RRKP:1}}.

\KH{}{The remainder of this paper is organized as follows. 
The system model is described in Section~\ref{sec:model}. 
Results for broadcast in one- and two-dimensional networks are presented in Sections~\ref{sec:1d} and~\ref{sec:2d}\KH{,}{} respectively. 
Multicast traffic is considered in Section~\ref{sec:mc}. 
Section~\ref{sec:conclusion} concludes the paper. }

\section{System Model}
\label{sec:model}

\KH{Consider a simple model for wireless networks where two stations have a direct radio link between them if they can hear each other.
The network can be modeled by an arbitrary graph $G=(V,A)$, where $V=\lbrace\mathbf{r}_i\rbrace_{i=0}^{\lvert V\rvert-1}$ is the set of stations labeled by their coordinates, and $A=\lbrace(\mathbf{r}_i,\mathbf{r}_j)\rbrace\subset V\times V$ is the set of \textit{undirected} links.\KHH{}{As an example, if $d(\mathbf{r}_i,\mathbf{r}_j)$ is the Euclidean distance between $\mathbf{r}_i$ and $\mathbf{r}_j$, and $R$ is the transmission range, then we can define the set of links to be $A_R=\lbrace(\mathbf{r}_i,\mathbf{r}_j)\in N\times N\colon i\ne j\text{ and }d(\mathbf{r}_i,\mathbf{r}_j)\le R\rbrace$.
In this case the graph, denoted by $G_R=(N,A_R)$, is a finite undirected \textit{unit disk graph}.}}{
Consider a simple model for wireless networks where two stations have a direct radio link between them if their distance is within a given range.
Precisely, the network is modeled by a finite undirected \textit{unit disk graph} $G_R=(N,A_R)$, where $N=\lbrace\mathbf{r}_i\rbrace_{i=0}^{\lvert N\rvert-1}$ is the set of stations labeled by their coordinates, $A_R=\lbrace(\mathbf{r}_i,\mathbf{r}_j)\in N\times N\colon i\ne j\text{ and }d(\mathbf{r}_i,\mathbf{r}_j)\le R\rbrace$ is the set of undirected links, $d(\mathbf{r}_i,\mathbf{r}_j)$ denotes the Euclidean distance between $\mathbf{r}_i$ and $\mathbf{r}_j$, and $R$ is the transmission range.}
Let $V_\mathbf{r}$ denote the set of (one-hop) peers or neighbors of station $\mathbf{r}$. 
\khh{We assume the interference range of a station is the same as its transmission range, so $V_\mathbf{r}$ denotes both the set of potential receivers and interferers for station $\mathbf{r}$.}{}
\KHH{}{In }Sections~\ref{sec:1d} and~\ref{sec:2d} \KHH{study the case where}{, it is assumed that} every station broadcasts packets to all its one-hop peers \KHH{in a single channel}{, \textit{i.e.}, station $\mathbf{r}$ transmits packets to $N_\mathbf{r}$}. 
\KHH{}{In }Section~\ref{sec:mc} \KHH{studies the case where}{,} every station multicasts packets to a \KHH{certain}{fixed} subset of its one-hop peers \KHH{in a single channel}{, \textit{i.e.}, station $\mathbf{r}$ transmits packets to $D_\mathbf{r}\subseteq N_\mathbf{r}$}.\KHH{}{It is assumed that station $\mathbf{r}$ knows $D_{\mathbf{r}^\prime}$ for all $\mathbf{r}^\prime\in N_\mathbf{r}$, so that it knows whether it is an intended receiver of its peers. 
This can be done by letting each station broadcast a list of its intended receivers while setting up a multicast session. }
\KHH{In Section~\ref{sec:multich}, broadcast and multicast in networks with multiple orthogonal channels are considered.}{}
For both broadcast and multicast, saturated traffic is assumed.

\begin{figure}[t]
\centering
\subfigure[]{
\includegraphics[width=3.5in]{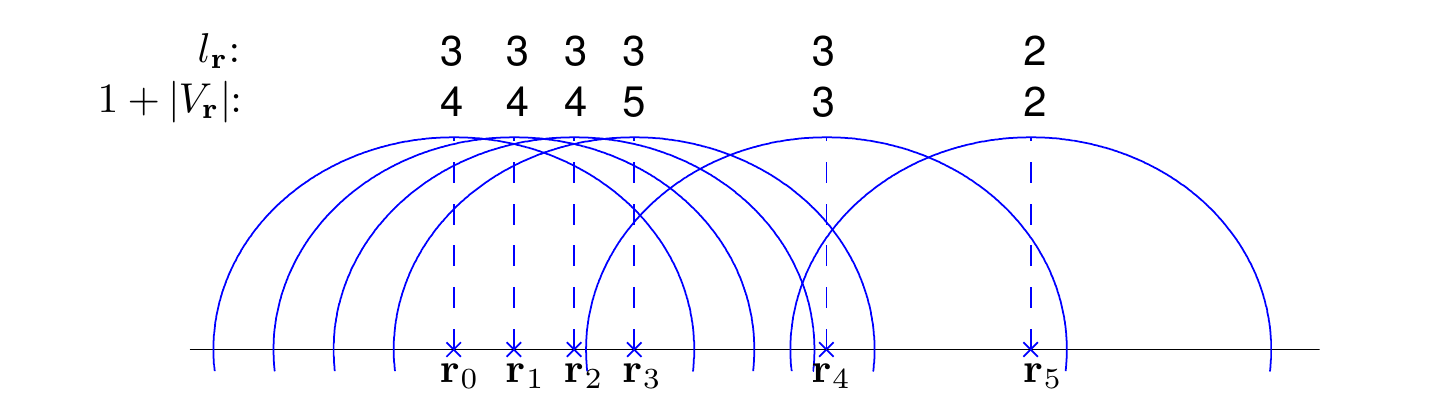}
\label{fig:network}
}
\subfigure[]{
\includegraphics[width=3.5in]{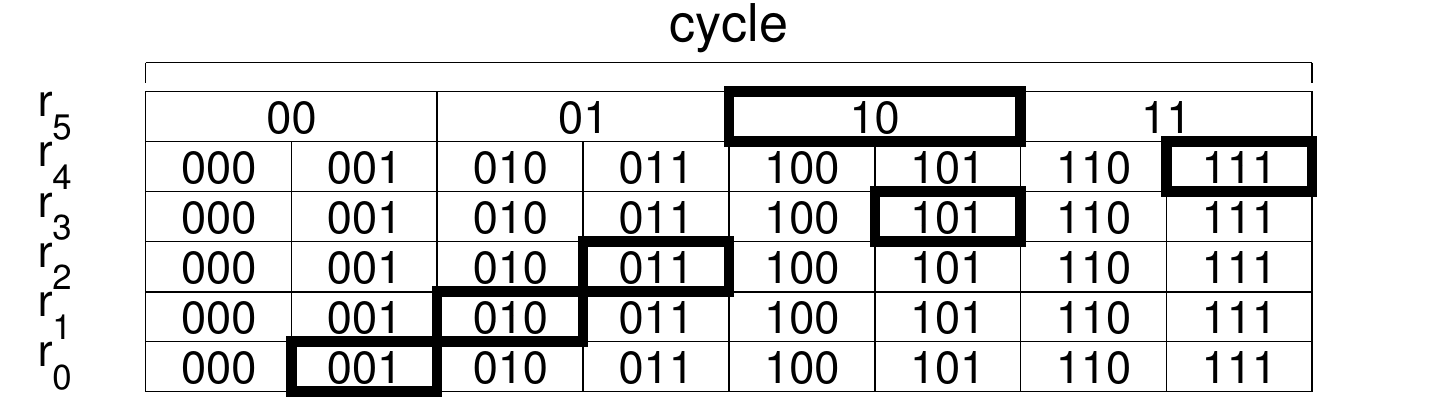}
\label{fig:resolution}
}
\caption{A multi-resolution MAC protocol in a one-dimensional network. In \subref{fig:network}, there are $6$ stations at positions $\mathbf{r}_0,\dots,\mathbf{r}_5$. The \khh{half}{} circles represent the transmission ranges of the stations. The values of $1+\lvert V_\mathbf{r}\rvert$ and $l_\mathbf{r}$ for different stations are shown on top of the corresponding circles. A possible schedule over a cycle is shown in \subref{fig:resolution}.}
\label{fig:mr}
\end{figure}

Next we formalize the concept of multiple resolutions\K{}{ alluded to above}. 
\KHH{
Let time be divided into cycles of fixed length. 
We let station $\mathbf{r}$ decide on a resolution represented by an integer $l_\mathbf{r}\ge0$. 
From the viewpoint of this station, each cycle is further divided into $2^{l_\mathbf{r}}$ slots of equal length. 
The state of this station in the $t$-th cycle, denoted by $X_\mathbf{r}(t)$, is a binary string of length $l_\mathbf{r}$ corresponding to the index of the slot in the $t$-th cycle over which the station transmits. 
Let $\pmb{X}(t)=\lbrace X_\mathbf{r}(t)\rbrace_{\mathbf{r}\in V}$ be the configuration in the $t$-th cycle. 
We assume that all stations are synchronized. 
\khh{A finer}{Therefore, a fine} resolution can be obtained by `splitting' or `refining' a coarse resolution. 
We assume that packets transmitted by a station fit in a slot of its own resolution. 
Stations using coarse resolutions can also transmit multiple packets of smaller sizes in a slot. 

A \textit{collision} occurs \K{between two}{at some receiver when two stations that are} one-hop or two-hop peers \K{if they}{} transmit at the same time. 
Mathematically, \K{two such}{neighboring} stations $\mathbf{r}_i$ and $\mathbf{r}_j$, with $l_{\mathbf{r}_i}\le l_{\mathbf{r}_j}$ \K{(without loss of generality)}{}, collide in the $t$-th cycle when 
\begin{equation}
\label{eqn:collision}
\text{the binary string }X_{\mathbf{r}_i}(t)\text{ is a prefix of }X_{\mathbf{r}_j}(t).
\end{equation}
In a \textit{collision-free configuration}, (\ref{eqn:collision}) \K{must}{does} not hold for any pair of one-hop and two-hop peers.}{
\KH{
Time is divided into cycles of fixed length. 
We let each station pick a state represented using a binary string $b_1b_2\cdots b_m$, which is the index of the time slot if a cycle is divided into $2^m$ slots.
Two packets arriving at a station collide if and only if the state of one of the two transmitting stations is a prefix of the state of the other.}{}}
Consider \KH{the example described in}{} Fig.~\ref{fig:mr}. 
Station $\mathbf{r}_5$ uses\KH{}{ a coarse resolution consisting of four states, $00,01,10,11$. 
It transmits in} state $10$, \textit{i.e.}, \KH{it transmits during}{} the third quarter of the cycle. 
Station $\mathbf{r}_3$ uses \KH{state $101$ with}{} a finer resolution consisting of eight states\KH{, so that it}{. 
It} transmits in the sixth slot of the cycle \KH{which is divided into $8$ slots}{when its state is $101$}. 
Station $\mathbf{r}_3$'s resolution can be seen as a refinement of that of station $\mathbf{r}_5$\KHH{}{, where every slot is reduced by half}.\KHH{}{Formally, the state of station $\mathbf{r}$ assumes values in $\mathbb{F}_2^{l_\mathbf{r}}$, the set of binary $l_\mathbf{r}$-tuples, where $2^{l_\mathbf{r}}$ is the \KH{total number of states available for}{number of states used by} station $\mathbf{r}$ \KH{to choose}{} and it depends on its local topology. 
Let $\Omega=\prod_{\mathbf{r}\in N}\mathbb{F}_2^{l_\mathbf{r}}$ denote the configuration space.
Let $X_\mathbf{r}(t)\in\mathbb{F}_2^{l_\mathbf{r}}$ be the state of station $\mathbf{r}$ in the $t$-th cycle, and $\pmb{X}(t)=\lbrace X_\mathbf{r}(t)\rbrace_{\mathbf{r}\in N}\in\Omega$ be the configuration in the $t$-th cycle. 
$X_\mathbf{r}(t)=s\in\mathbb{F}_2^{l_\mathbf{r}}$ means that station $\mathbf{r}$ divides the $t$-th cycle into $2^{l_{\mathbf{r}}}$ slots of equal length, and it transmits in slot $s$. 
For this reason, we will use states and slots interchangeably. 
We assume that all stations are synchronized. 
Therefore, a fine resolution can be obtained by `splitting' or `refining' a coarse resolution, as shown in Fig.~\ref{fig:resolution}. }
\KHH{Since $10$ is a prefix of $101$, and stations $\mathbf{r}_3$ and $\mathbf{r}_5$ are two-hop peers as shown in Fig.~\ref{fig:network}, these two stations collide \K{(at receiver $\mathbf{r}_4$)}{}.}{}

\KHH{}{
We assume that packets transmitted by a station fit in a slot of its own resolution. 
Stations using coarse resolution can also transmit multiple packets of smaller sizes in a slot. 
A \textit{collision} occurs at some receiver when two stations that are one-hop or two-hop peers transmit at the same time. 
Mathematically, neighboring stations $\mathbf{r}_i$ and $\mathbf{r}_j$, with $l_{\mathbf{r}_i}\le l_{\mathbf{r}_j}$, collide in the $t$-th cycle when 
\begin{equation}
\label{eqn:collision}
\text{the binary string }X_{\mathbf{r}_i}(t)\text{ is a prefix of }X_{\mathbf{r}_j}(t).
\end{equation}
For example, in Fig.~\ref{fig:resolution}, since $10$ is a prefix of $101$, stations $\mathbf{r}_3$ and $\mathbf{r}_5$ collide.
In a \textit{collision-free configuration}, (\ref{eqn:collision}) does not hold for any pair of one-hop and two-hop peers.
}

\KH{We assume that at the end of each cycle, each station acquires the current states of its one-hop and two-hop peers, error-free.
\KHH{Such message exchanges can be carried out either over a control channel or over a dedicated time period.}{}
The careful reader may object that this itself requires a collision-free schedule.
However, since this control information is relatively low-rate, we assume that other techniques can be utilized for sending it.
For example, \KHH{stations can use a random access scheme to exchange the short control messages.}{}
\kh{Alternatively, the}{\KHH{T}{t}he} rapid on-off division duplex (RODD) scheme in \cite{DGLZ:1,LZDG:1} can\kh{}{ also} be used here, which enables all stations to exchange their control messages simultaneously.
\KHH{From now on, we will assume stations exchange state information within a control frame orthogonal to data frames in time or in frequency.\footnote{\K{Assuming that the control frame is short,}{The control frame is typically short. We ignore} its impact on the throughput \K{is ignored}{} in this paper.}}{}
}{
We assume that at the end of each cycle, each station acquires the current states of its one-hop and two-hop peers.
Here, every station broadcasts a message to all its one-hop peers, and tries to receive a message from every peer at the same time. 
\KH{
Such message exchanges can be carried out either over a control channel or over a dedicated time period. 
It is conceivable to use a random access scheme for exchanging the short state messages.}{}
\KH{There}{Though wireless systems are half-duplex, there} is a recent work on achieving virtual full-duplex communication in wireless systems \KH{using half duplex radio}{}, called rapid on-off division duplex (RODD) \cite{DGLZ:1}. 
In this scheme, all stations can exchange a message with their respective one-hop peers within one frame interval.  
Each station is assigned an on-off duplex mask of one frame length.
In an on-slot of the frame, the station transmits a symbol; whereas in an off-slot it does not emit any energy and therefore can receive a signal. 
As long as the masks are sufficiently different,  a station can receive enough signal through its off-slots and decode messages from its peers. 
\KH{From now on}{Hence}, we will assume stations exchange state information\KH{}{ by techniques like RODD} within a control frame \KH{orthogonal to data frames in time or in frequency}{}.\footnote{The control frame is typically short. We ignore its impact on the throughput in this paper.}}

Let stations choose their next states based only on the current states of their one-hop and two-hop peers and themselves.
The state process of the MAC protocol can be modeled as a \textit{Markov Chain of Markov Fields} (MCMF) \cite{XGCH:1}, \textit{i.e.}, a process for which the states $\pmb{X}=\lbrace\pmb{X}(t)\rbrace_{t\in\mathbb{N}}$ satisfy
\begin{itemize}
\item
$\pmb{X}(1),\pmb{X}(2),\dots$ is a Markov chain\KHH{}{ on $\Omega$}, and
\item
for every $t$, $\pmb{X}(t)$ is a Markov field\KHH{}{ on $\Omega$} conditioned on $\pmb{X}(t-1)$.
\end{itemize}
\KH{In fact, $\pmb{X}(t)$ consists of independent random variables conditioned on $\pmb{X}(t-1)$ in our case.}{}
Here, we only consider protocols in which stations make identically distributed decisions conditioned on the same previous states of their one-hop and two-hop peers and themselves.\footnote{This rules out \K{location-based MAC protocols}{protocols in which stations, for example, are simply assigned to transmit or not based on their location} (\textit{e.g.}, in \cite{NWRB:2}).}

In Sections~\ref{sec:1d} and~\ref{sec:2d} we measure the performance by the one-hop broadcast throughput $\rho_\text{BC}$, which is the average proportion of time a station receives packets in each cycle. 
A station receives a packet if and only if it does not transmit and only one of its peers transmits. 
If there is no collision, 
\KHH{\begin{equation}
\label{eqn:rho}
\rho_\text{BC}=\frac{1}{\lvert V\rvert}\sum_{\mathbf{r}\in V}\sum_{\mathbf{r}^\prime\in V_\mathbf{r}}2^{-l_{\mathbf{r}^\prime}}=\frac{1}{\lvert V\rvert}\sum_{\mathbf{r}\in V}\lvert V_\mathbf{r}\rvert2^{-l_\mathbf{r}}.
\end{equation}
}{\begin{equation}
\label{eqn:rho}
\textstyle\rho_\text{BC}=\bigl\langle\sum_{\mathbf{r}^\prime\in N_\mathbf{r}}2^{-l_{\mathbf{r}^\prime}}\bigr\rangle_\mathbf{r}=\bigl\langle\lvert N_\mathbf{r}\rvert2^{-l_\mathbf{r}}\bigr\rangle_\mathbf{r}
\end{equation}
where $\langle\cdot\rangle_\mathbf{r}$ in (\ref{eqn:rho}) is the spatial average over all stations, \textit{i.e.}, $\langle g(\mathbf{r})\rangle_\mathbf{r}=\frac{1}{\lvert N\rvert}\sum_{\mathbf{r}\in N}g(\mathbf{r})$ for any function $g$.}The two expressions are obtained by counting throughput from the receiver side and the transmitter side\K{,}{} respectively. 
In Section~\ref{sec:mc}, we use the one-hop multicast throughput \KH{$\rho_\text{MC}$}{} to measure the performance. 
In this case, a station receives a packet if and only if it does not transmit, only one of its peers transmits and it is an intended receiver for the packet. 
\KHH{Let $D_\mathbf{r}\subseteq N_\mathbf{r}$ denote the set of intended receivers of the multicast by $\mathbf{r}$.}{}
\KHH{If there is no collision, the one-hop multicast throughput is,}{The one-hop multicast throughput is, if there is no collision,} 
\KHH{\begin{equation}
\label{eqn:rho_mc}
\rho_\text{MC}=\frac{1}{\lvert V\rvert}\sum_{\mathbf{r}\in V}\sum_{\mathbf{r}^\prime\colon\mathbf{r}\in D_{\mathbf{r}^\prime}}2^{-l_{\mathbf{r}^\prime}}=\frac{1}{\lvert V\rvert}\sum_{\mathbf{r}\in V}\lvert D_\mathbf{r}\rvert2^{-l_\mathbf{r}}.
\end{equation}}{\begin{equation}
\label{eqn:rho_mc}
\textstyle\rho_\text{MC}=\bigl\langle\sum_{\mathbf{r}^\prime\colon\mathbf{r}\in D_{\mathbf{r}^\prime}}2^{-l_{\mathbf{r}^\prime}}\bigr\rangle_\mathbf{r}=\bigl\langle\lvert D_\mathbf{r}\rvert2^{-l_\mathbf{r}}\bigr\rangle_\mathbf{r}.
\end{equation}}

\KHH{It should be noted that under the concept of multiple resolutions, the structure of the states can be more complex than that described here. 
For example, the states may be represented by tertiary codes, so the number of slots in a cycle need not be a power of $2$. 
Also, to represent collisions using the prefix condition (\ref{eqn:collision}), it is not required that all slots in a cycle must have the same length; the only requirements are that all slot boundaries of a coarse resolution are also slot boundaries of a fine resolution, and two slots overlap in time if and only if the states representing the slots satisfy the prefix condition. 
\K{}{The model described here is the simplest one, and is used for illustrative purposes only. }
}{}

\section{Broadcast in One-Dimensional Networks}
\label{sec:1d}

\subsection{Determining the number of states}
\label{subsec:l_1d}

\KH{
We first consider one-dimensional networks, \textit{i.e.}, all stations lie on a straight line.
We further assume the following: if $\mathbf{r}_i$ and $\mathbf{r}_j$ are one-hop peers, then all stations located between $\mathbf{r}_i$ and $\mathbf{r}_j$ are also one-hop peers of both $\mathbf{r}_i$ and $\mathbf{r}_j$.\KHH{}{This assumption is valid since stations closer to a transmitter receive stronger signals.}\K{}{In this situation, s}}{S}\K{}{tation $\mathbf{r}$ determines its resolution\KHH{}{ (in bits)} $l_\mathbf{r}$ as follows:
\begin{enumerate}
\item
station $\mathbf{r}$ computes $w_\mathbf{r}=1+\lvert V_\mathbf{r}\rvert$, the number of stations within $\mathbf{r}$'s one-hop neighborhood, and exchanges this with all its one-hop peers (\textit{e.g.}, using techniques of \cite{DGLZ:1}),
\item
station $\mathbf{r}$ \KH{sets}{uses} $l_\mathbf{r}=\lceil\log_2(\max_{\mathbf{r}^\prime\in V_\mathbf{r}\cup\lbrace \mathbf{r}\rbrace}w_{\mathbf{r}^\prime})\rceil$.
\end{enumerate}}
\KH{To avoid collision, a station and all its one-hop peers must transmit at different times\kh{}{, implying the size of the state space of a station should be at least equal to the size of}}{The intuition is that the number of states used by a station should be larger than}\kh{}{ the largest one-hop neighborhood \KH{that the station belongs to}{that the station belongs to}}. 
\K{The following result shows that a station can determine its resolution solely based on the size of the largest one-hop neighborhood that it belongs to.}{Fig.~\ref{fig:network} illustrates the procedure \K{of finding the resolutions from the following result}{}. }

\begin{thm}
\label{thm:num_states_1d}
\K{Suppose in a one-dimensional network, each station shares the number of stations within its one-hop neighborhood (\textit{i.e.}, $1+|V_\mathbf{r}|$ for station $\mathbf{r}$) with all its one-hop peers, then collision-free configurations are guaranteed to exist by letting stations choose their resolutions according to
\begin{equation}
\label{eqn:resolution_1d}
l_\mathbf{r}=\biggl\lceil\log_2\biggl(\max_{\mathbf{r}^\prime\in V_\mathbf{r}\cup\lbrace\mathbf{r}\rbrace}(1+\lvert V_{\mathbf{r}^\prime}\rvert)\biggr)\biggr\rceil.
\end{equation}
}{For any one-dimensional network, if station $\mathbf{r}$ uses $l_\mathbf{r}=\lceil\log_2(\max_{\mathbf{r}^\prime\in V_\mathbf{r}\cup\lbrace\mathbf{r}\rbrace}w_{\mathbf{r}^\prime})\rceil$, where $w_\mathbf{r}=1+\lvert V_\mathbf{r}\rvert$, then it is possible for each station to choose a state such that collision-free configurations exist.}\KH{The resulting one-hop broadcast throughput is \kh{given by (\ref{eqn:rho}), where $l_\mathbf{r}$ in (\ref{eqn:rho}) is specified in (\ref{eqn:resolution_1d}).}{
\KHH{\begin{equation}
\label{eqn:throughput_1d}
\rho_\text{BC}=\frac{1}{\lvert V\rvert}\sum_{\mathbf{r}\in V}\lvert V_\mathbf{r}\rvert2^{-\lceil\log_2(\max_{\mathbf{r}^\prime\in V_\mathbf{r}\cup\lbrace\mathbf{r}\rbrace}(1+\lvert V_{\mathbf{r}^\prime}\rvert))\rceil}.
\end{equation}}{\begin{equation}
\label{eqn:throughput_1d}
\rho_\text{BC}=\bigl\langle\lvert N_\mathbf{r}\rvert2^{-\lceil\log_2(\max_{\mathbf{r}^\prime\in N_\mathbf{r}\cup\lbrace\mathbf{r}\rbrace}(1+\lvert N_{\mathbf{r}^\prime}\rvert))\rceil}\bigr\rangle_\mathbf{r}.
\end{equation}}}}{}
\end{thm}

\K{
Fig.~\ref{fig:network} illustrates the procedure in Theorem~\ref{thm:num_states_1d}. 
Each station computes the size of its one-hop neighborhood (which is labeled on top of the \khh{half}{} circle representing its transmission range). 
Stations then choose their resolutions following (\ref{eqn:resolution_1d}). 
}{}

\K{}{\KH{We will introduce \KHH{a}{the} multi-resolution protocol in \KHH{Section~\ref{sec:multi_resolution_1d}, which has guaranteed convergence and achieves the throughput given in Theorem~\ref{thm:num_states_1d}}{the next subsection, and prove this result together with the convergence of such protocol}}{The proof of this result will be defered until we introduce the multi-resolution protocol, where the convergence of such protocol also proves this result}. }

\begin{figure}[t]
\centering
\includegraphics[width=3.5in]{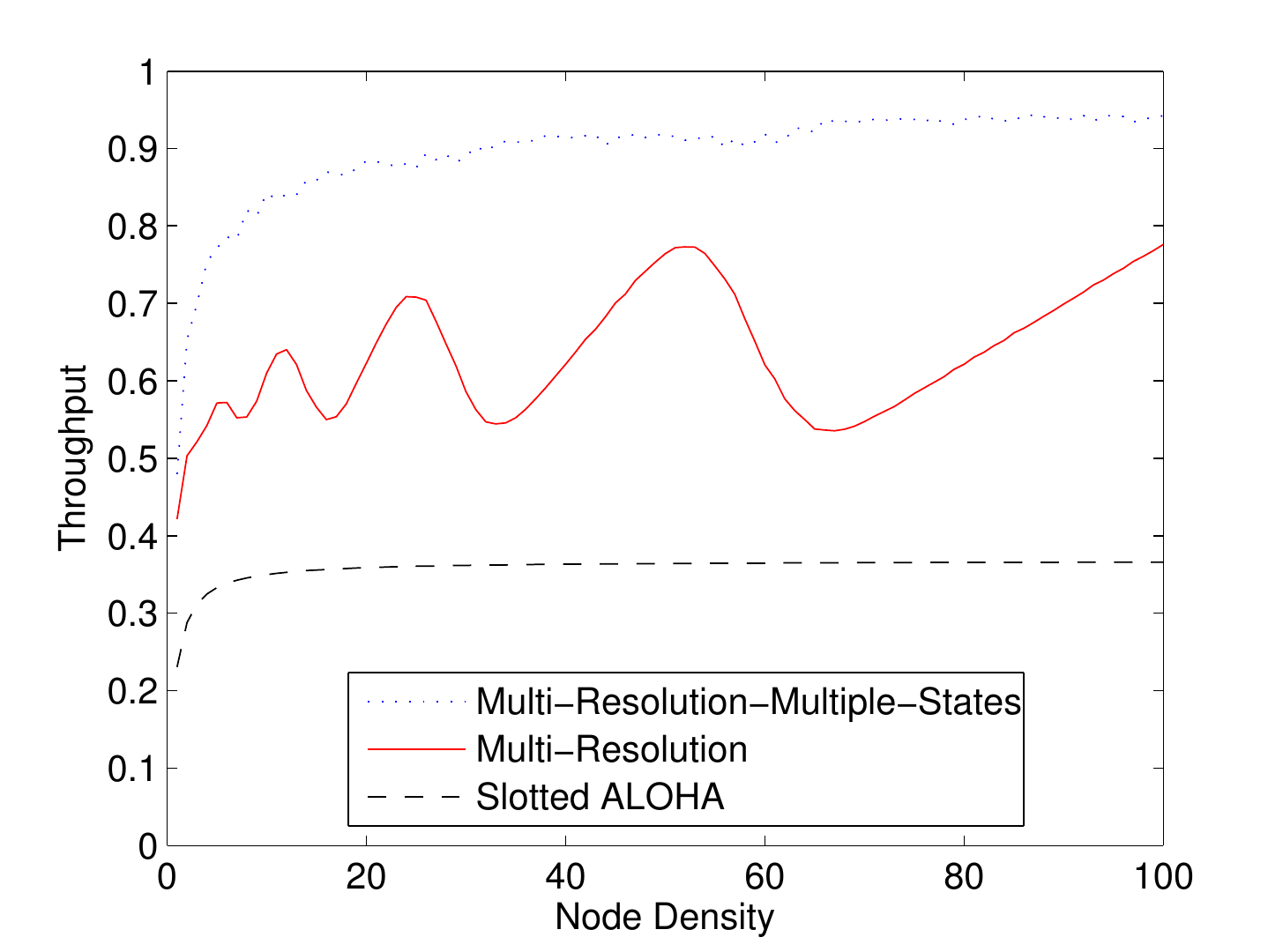}
\caption{Throughput of one-dimensional networks versus node \K{density}{intensity}.}
\label{fig:throughput}
\end{figure}

\KH{}{
\begin{thm}
\label{thm:throughput_1d}
Assume station $\mathbf{r}$ determines $l_\mathbf{r}$ following Theorem~\ref{thm:num_states_1d}, and each station transmits in one slot only. 
The one-hop broadcast throughput of an one-dimensional network is 
\begin{equation}
\label{eqn:throughput_1d}
\rho=\bigl\langle\lvert N_\mathbf{r}\rvert2^{-\lceil\log_2(\max_{\mathbf{r}^\prime\in N_\mathbf{r}\cup\lbrace\mathbf{r}\rbrace}(1+\lvert N_{\mathbf{r}^\prime}\rvert))\rceil}\bigr\rangle_\mathbf{r}.
\end{equation} 
\end{thm}}

\KH{
Consider finite segments of one-dimensional networks where stations are distributed following a Poisson point process with \K{node density}{intensity} $\lambda$\KHH{}{, and the transmission range is $R=1$}.
\KHH{We assume that the network is a \textit{unit disk graph} with transmission range $R=1$, \textit{i.e.}, there is a link between two stations if and only if the distance between them is at most equal to the transmission range $R=1$.}{}
We evaluate \kh{the throughput}{(\ref{eqn:throughput_1d})} by averaging over \khh{$100$}{} different realizations of the networks.
How the throughput $\rho_\text{BC}$ varies with the \K{node density}{intensity} $\lambda$ is shown in Fig.~\ref{fig:throughput}.}{
Consider a finite segment of an one-dimensional network where stations are distributed following a Poisson point process with intensity $\lambda$. 
The transmission range is $R=1$. 
How the throughput $\rho$ varies \KH{with}{against} the intensity $\lambda$ is shown in Fig.~\ref{fig:throughput}. }
The throughput oscillation is due to the fact that \KH{the size of any state space is}{all number of states have to be} a power of $2$. 
In a worst-case scenario, if a station determines that it needs $2^l+1$ states, it\KHH{}{ actually} has to use \KH{a resolution of}{} $2^{l+1}$ states, meaning that\KHH{}{ potentially} almost half of the cycle will be left idle, hence the throughput \KHH{is}{can be very} close to $0.5$. 
\KH{\KHH{Thus}{Hence}, a small increase in the \K{node density}{intensity} may cause a rather large \KHH{drop in the throughput \K{under certain circumstances}{}}{fraction of slots to be wasted}.}{}

After forming a collision-free configuration, there may still be many idle slots \KH{in certain neighborhoods}{such that stations can transmit in these slots without collision, thereby increasing the throughput}. 
\K{
To illustrate this, consider a collision-free configuration which is formed by letting stations, from the left to the right, pick the earliest slot to transmit such that they do not collide with any station. 
We \khh{then}{} let stations, from the left to the right, reclaim the idle slots to transmit, such that station $\mathbf{r}$ reclaims at most $\Bigl\lceil\frac{2^{l_\mathbf{r}}}{\lvert V_\mathbf{r}\rvert}\Bigr\rceil-1$ additional slots and ensures that it does not collide with other stations. 
By doing so, station $\mathbf{r}$ transmits at a rate approximately equal to $\frac{1}{\lvert V_\mathbf{r}\rvert}$. 
\khh{The top curve in}{} Fig.~\ref{fig:throughput} shows that a significant improvement in throughput \khh{results from this reclaiming}{is resulted by letting stations reclaim the idle slots to transmit}. 
}{
How stations \KH{reclaim}{transmit in} these idle slots is left to future work. 
}

For comparison, we also compute the throughput for slotted ALOHA \KH{in a one-dimensional network, where stations use the same fixed transmission probability $p$ but do not exchange any state information}{}. 
Consider a segment of a one-dimensional network of length $2R$ with a station at the center. 
This station has $k$ peers with probability $\exp(-\lambda2R)\frac{(\lambda2R)^k}{k!}$\khh{, and}{.\KH{}{Assume stations transmit with a fixed probability $p$.  }
This station} receives a packet successfully with probability $kp(1-p)^k$. 
Then, 
\begin{IEEEeqnarray}{rCl}
\label{eqn:sa_p}
\rho_\text{BC}(p)&=&\sum_{k=1}^\infty\exp(-\lambda2R)\frac{(\lambda2R)^k}{k!}kp(1-p)^k\IEEEnonumber\\
&=&\lambda2Rp(1-p)\exp(-\lambda2Rp).\IEEEnonumber
\end{IEEEeqnarray}
\KH{The maximum throughput is}{Optimizing over all $p$, the throughput is}
\begin{displaymath}
\rho_\text{BC}=\frac{\lambda2R}{2+\sqrt{4+(\lambda2R)^2}}\exp\Biggl(-\frac{2\lambda2R}{2+\lambda2R+\sqrt{4+(\lambda2R)^2}}\Biggr)
\end{displaymath}
\K{which is achieved with transmission probability}{\KH{if the transmission probability is chosen to be}{with the corresponding $p$ given by}}
\begin{displaymath}
\label{eqn:p_opt_sa}
p=\frac{2}{2+\lambda2R+\sqrt{4+(\lambda2R)^2}}.
\end{displaymath}
This optimized throughput\KH{, with $R=1$,}{} is\KH{}{ also} plotted in Fig.~\ref{fig:throughput}. 
The multi-resolution MAC protocol provides $46.7\%$ to $112.2\%$ improvement in terms of throughput over slotted ALOHA. 

\subsection{\kh{}{A }Multi-Resolution MAC Protocol}
\label{sec:multi_resolution_1d}

In the following we propose a \textit{multi-resolution protocol} that leads to a collision-free configuration \textit{starting from an arbitrary initial configuration}. 
Stations can learn two-hop state information in each cycle as follows. 
In the $t$-th cycle, station $\mathbf{r}$ collects $\bigl\lbrace X_{\mathbf{r}^\prime}(t)\bigr\rbrace_{\mathbf{r}^\prime\in V_\mathbf{r}}$, and then broadcasts $\bigl\lbrace X_{\mathbf{r}^\prime}(t)\bigr\rbrace_{\mathbf{r}^\prime\in V_\mathbf{r}\cup\lbrace\mathbf{r}\rbrace}$. 
Hence, station $\mathbf{r}$ knows \KH{$X_{\mathbf{r}^\prime}(t)$ for all one-hop and two-hop peers $\mathbf{r}^\prime$}{$\bigl\lbrace X_{\mathbf{r}^{\prime\prime}}(t)\bigr\rbrace_{\mathbf{r}^{\prime\prime}\in N_{\mathbf{r}^\prime}\cup\lbrace\mathbf{r}^\prime\rbrace}$ for all $\mathbf{r}^\prime\in N_\mathbf{r}\cup\lbrace\mathbf{r}\rbrace$} (this is accomplished by letting station $\mathbf{r}$ broadcast $2l_\mathbf{r}+\sum_{\mathbf{r}^\prime\in V_\mathbf{r}}l_{\mathbf{r}^\prime}$ bits), and selects its state at the $(t+1)$-st cycle following Protocol~\ref{alg:mr_bc}\kh{, where the parameter $\epsilon$ is set to 0 in the case of one-dimensional networks (in case of two-dimensional networks discussed in Section~\ref{sec:2d}, we will set $\epsilon$ to a strictly positive number)}{}. 

\K{
\begin{algorithm}[t]
\caption{Multi-Resolution MAC Protocol for Broadcast}
\label{alg:mr_bc}
\begin{algorithmic}[1]
\WHILE{\khh{station $\mathbf{r}$ is active}{}}
\STATE
$\mathbf{r}$ sets the votes on all states to zero. 
\FOR{$\mathbf{r}^\prime\in V_\mathbf{r}\cup\lbrace\mathbf{r}\rbrace$}
\IF{$\mathbf{r}$ is the only station occupying its current state in station $\mathbf{r}^\prime$'s one-hop neighborhood} 
\STATE
\KH{$\mathbf{r}$'s current state is assigned a single vote of weight one}{A single vote of weight one on $\mathbf{r}$'s current state is given by $\mathbf{r}^\prime$}. 
\ELSE
\STATE
\KH{$\mathbf{r}$ determines which states (according to $\mathbf{r}$'s resolution) are idle or have collisions in $\mathbf{r}^\prime$'s one-hop neighborhood}{$\mathbf{r}$ determines the states (according to $\mathbf{r}$'s resolution) that station $\mathbf{r}^\prime$ is idle or collides}. 
\STATE
A vote of weight $\frac{1}{n}$ is \khh{added}{given} to each such state\KH{}{ by station $\mathbf{r}^\prime$}, where $n$ is the number of such states. 
\ENDIF
\ENDFOR
\IF{$n_s>0$ for multiple $s$'s, where $n_s$ is the total weight state $s$ receives}
\STATE
Replace $n_s$ by $n_s+\epsilon$, where $\epsilon\ge0$, for all $s$. 
\ENDIF
\STATE
$\mathbf{r}$ selects state $s$ with a probability proportional to $f(n_s)$.
\ENDWHILE
\end{algorithmic}
\end{algorithm}
}{
\begin{algorithm}[t]
\caption{Multi-Resolution MAC Protocol \KHH{for Broadcast on One-Dimensional Networks}{}}
\label{alg:mr_bc_1d}
\begin{algorithmic}[1]
\STATE
$\mathbf{r}$ sets the votes on all states to zero. 
\FOR{$\mathbf{r}^\prime\in V_\mathbf{r}\cup\lbrace\mathbf{r}\rbrace$}
\IF{$\mathbf{r}$ is the only station occupying its current state in station $\mathbf{r}^\prime$'s one-hop neighborhood} 
\STATE
\KH{$\mathbf{r}$'s current state is assigned a single vote of weight one}{A single vote of weight one on $\mathbf{r}$'s current state is given by $\mathbf{r}^\prime$}. 
\ELSE
\STATE
\KH{$\mathbf{r}$ determines which states (according to $\mathbf{r}$'s resolution) are idle or have collisions in $\mathbf{r}^\prime$'s one-hop neighborhood}{$\mathbf{r}$ determines the states (according to $\mathbf{r}$'s resolution) that station $\mathbf{r}^\prime$ is idle or collides}. 
\STATE
A vote of weight $\frac{1}{n}$ is given to each such state\KH{}{ by station $\mathbf{r}^\prime$}, where $n$ is the number of such states. 
\ENDIF
\ENDFOR
\STATE
$\mathbf{r}$ selects state $s$ with a probability proportional to $f(n_s)$, where $n_s$ is the total weight state $s$ receives.
\end{algorithmic}
\end{algorithm}
}

In Protocol~\ref{alg:mr_bc}, $f:\mathbb{R}\mapsto\mathbb{R}$ \KH{can be any}{is an} increasing function with $f(0)=0$\KH{. Empirically\khh{,}{} a good choice is }{, \textit{e.g.},} $f(n_s)=\exp(Jn_s)\mathbf{1}_{\lbrace n_s>0\rbrace}$, where $\mathbf{1}_{\lbrace\cdot\rbrace}$ is the indicator function and \KH{$J>0$}{$J$} is the \textit{strength of interaction} (\KHH{more on this}{discussed} later). 
The idea of Protocol~\ref{alg:mr_bc} is that a station `reserves' a slot \KH{for}{to} a peer if it knows that this peer does not collide with other peers, and notifies any peer experiencing collisions to stay away from these `reserved' slots. 
We have the following convergence result for this protocol. 

\begin{thm}
\label{thm:mr_1d}
If each station in a one-dimensional network chooses its resolution following Theorem~\ref{thm:num_states_1d} and executes Protocol~\ref{alg:mr_bc}, then \khh{all stations will converge to a collision-free configuration}{a collision-free configuration will be formed after a sufficiently long time}, regardless of \KH{the}{their} initial state\KH{}{s}. 
\K{The resulting throughput is given \kh{in Theorem~\ref{thm:num_states_1d}}{by (\ref{eqn:throughput_1d})}.}{}
\end{thm}

\begin{IEEEproof}
\KH{
\KHH{}{It suffices to show that starting from an arbitrary configuration, there is a sequence of events which takes place with nonzero probability so that the network evolves to a collision-free configuration, which is an absorption state so that the network stays collision-free ever since. }Using Protocol~\ref{alg:mr_bc}, if a station does not collide with any one-hop or two-hop peers, then all the votes will be given to its current state, and it will remain in its current state with probability one.}{
Notice that for the proposed protocol, every station must have nonzero probability in choosing the same state in two consecutive time slots: 
\begin{enumerate}
\item
If it does not collide with any one-hop or two-hop peers, then all the votes will be given to its current state, and it will remain in its current state with probability one. 
\item
Otherwise, any one-hop peer detecting this collision will give a vote of nonzero weight to its current state, meaning that it will remain in its current state with nonzero probability. 
\end{enumerate}
Inductively, any station can choose the same state for any number of consecutive time slots with nonzero probability. }
\KH{Therefore}{In particular}, if the \K{current}{initial} configuration is collision-free, then the same\khh{}{collision-free} configuration will appear in every subsequent cycle, so every collision-free configuration is absorbing. 
Hence we only need to consider the case when the \khh{current}{initial} configuration is not collision-free and show that \khh{such configuration is}{configurations with collision are} transient. 
\khh{To do this we explicitly construct a collision-free configuration, which the stations in the current configuration can transition to with positive probability.}{}

\KHH{}{We show next that\KH{}{ with nonzero probability} a collision-free configuration can be reached from an arbitrary configuration. }Without loss of generality, assume the stations are indexed such that $\mathbf{r}_i$ is on the left of $\mathbf{r}_j$ if and only if $i<j$. 
Stations take turn\KH{s}{} to find a state that is collision-free with all stations on \KHH{their}{its} left: 
\begin{itemize}
\item
Station $\mathbf{r}_0$ remains in its initial state, so it is collision-free with all stations on its left \KH{(notice that following Protocol~\ref{alg:mr_bc}, \textit{every station has a nonzero probability of remaining in the same state})}{}. 
\item
Now, assume stations $\mathbf{r}_0,\dots,\mathbf{r}_{i-1}$ are collision-free with all stations on their left. For station $\mathbf{r}_i$:
\begin{enumerate}
\item
If its current state is collision-free with all stations on its left (including the special case where there is no neighboring station on its left), then it remains in its current state. 
\item
Otherwise, consider the farthest left one-hop peer of station $\mathbf{r}_i$, which we denote by $\mathbf{r}_j$. 
If $\mathbf{r}_i$ collides with some station $\mathbf{r}_k$ on its left, then $\mathbf{r}_j$ must be able to detect it, because $\mathbf{r}_j$ must be one-hop peers of both $\mathbf{r}_i$ and $\mathbf{r}_k$ ($\mathbf{r}_j$ and $\mathbf{r}_k$ can be the same station). 
$\mathbf{r}_j$ and all one-hop peers $\mathbf{r}_m$ of $\mathbf{r}_j$ use \KHH{resolutions of}{} at least $2^{\lceil\log_2(1+\lvert V_{\mathbf{r}_j}\rvert)\rceil}$ states.
Therefore, from station $\mathbf{r}_j$'s point of view, there are at most $\lvert V_{\mathbf{r}_j}\rvert$ distinct busy periods, each of length at most $2^{-\lceil\log_2(1+\lvert V_{\mathbf{r}_j}\rvert)\rceil}$. 
This means that there is at least one idle slot according to $\mathbf{r}_j$'s resolution, \textit{i.e.}, none of the $\mathbf{r}_m$'s use that slot. 
Therefore, $\mathbf{r}_j$ gives a vote of nonzero weight on this slot to $\mathbf{r}_i$, then with nonzero probability, $\mathbf{r}_i$ chooses this slot (or a fraction of this slot if it uses a finer resolution) and becomes collision-free with all stations on its left. 
\end{enumerate}
\end{itemize}
Finally, when station $\mathbf{r}_{\lvert V\rvert-1}$ finds a state that is collision-free with all stations on its left, the configuration is now collision-free. 
Therefore, all configurations with collision\khh{s}{} are transient\khh{, proving both Theorems~\ref{thm:num_states_1d} and~\ref{thm:mr_1d}}{. 

The above construction also proves Theorem~\ref{thm:num_states_1d}}. 
\end{IEEEproof}

\subsection{Simulations: Convergence Speed-up by Annealing}
\label{subsec:annealing_1d}

\begin{figure}[t]
\centering
\subfigure[Convergence time]{
\includegraphics[width=3.5in]{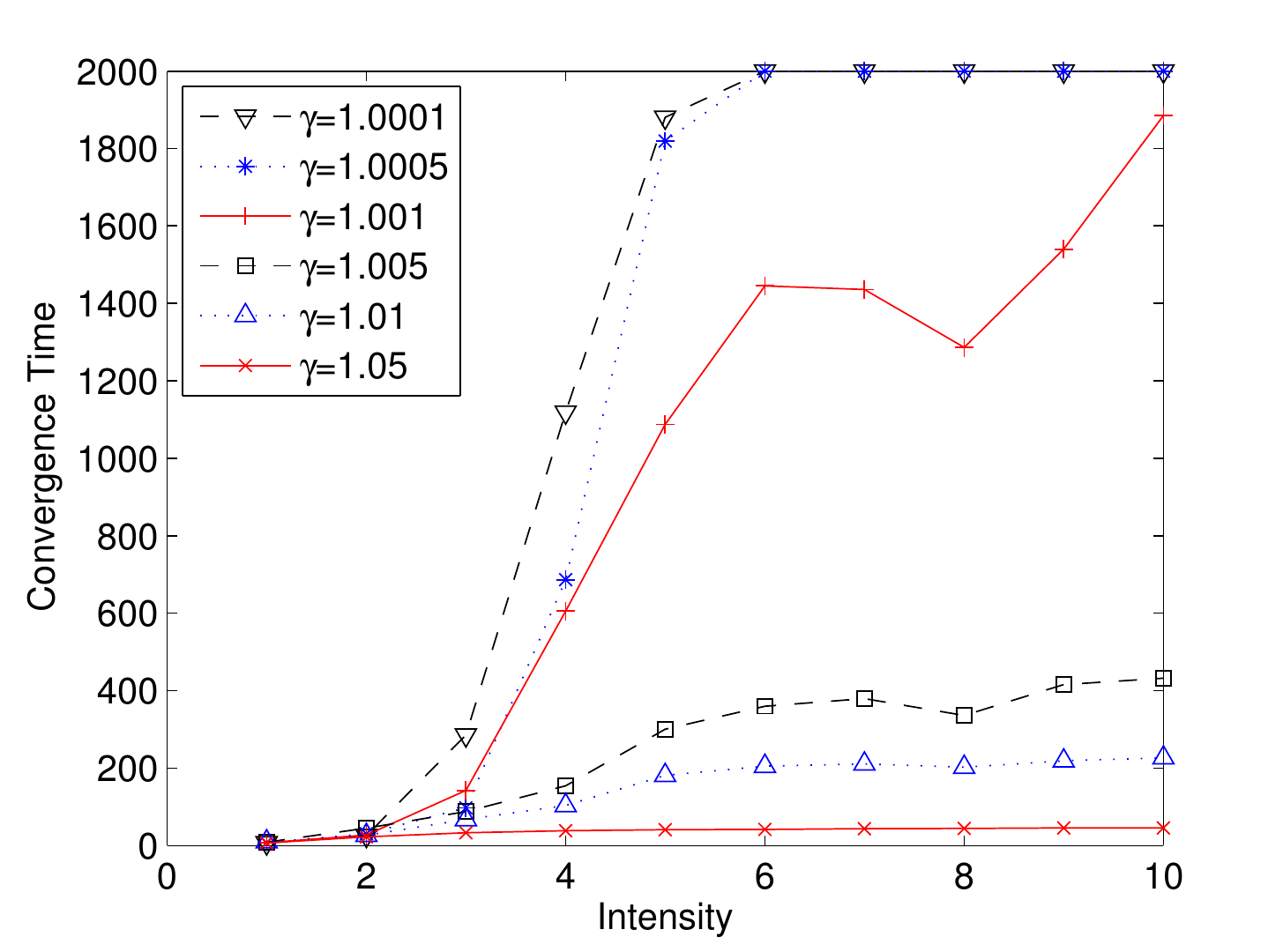}
\label{fig:time_1d}
}
\hspace{-0.4in}
\subfigure[Convergence percentage]{
\includegraphics[width=3.5in]{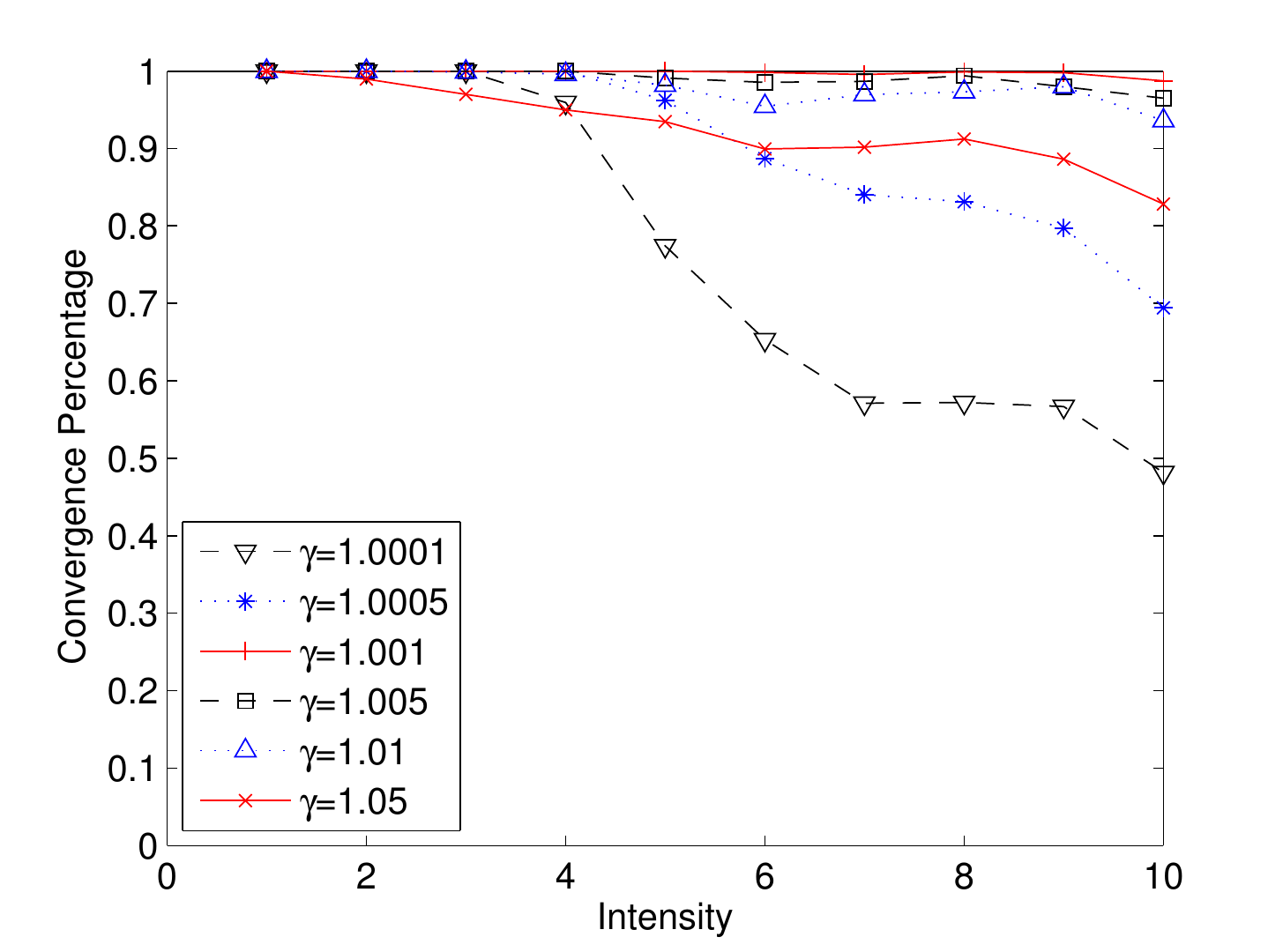}
\label{fig:percentage_1d}
}
\caption[]{Simulations of the multi-resolution protocol with annealing for one-dimensional networks.}
\label{fig:annealing_1d}
\end{figure}

\KH{S}{Preliminary s}imulations of the proposed protocol show\KH{}{s} that it may take a long time for a collision-free configuration to appear. 
Here we propose speeding up the convergence by \textit{annealing}, \textit{i.e.}, we consider the multi-resolution protocol with $f(n_s)=\exp(J(t)n_s)\mathbf{1}_{\lbrace n_s>0\rbrace}$, where $J(t)=\gamma J(t-1)$, $\gamma>1$ controls the increase in the strength of interaction, and $J(0)=1$.
Define the convergence time to be the first time that a \KH{certain}{fixed} configuration is observed \KH{and remains unchanged till the end of the simulation}{}, and the convergence percentage to be the proportion of stations that do not collide with other stations \KH{in that configuration}{when the fixed configuration is observed}.
\textit{\KH{This}{Notice that this fixed} configuration may not be collision-free}. 
\textit{This means that there is\KH{}{ still} a nonzero probability that the network transits to another configuration, but this probability is \KH{so}{too} small (as $J(t)$ is very large, resulting in every station staying in the state with maximum vote) that this transition is practically impossible.}
We consider \KH{a line segment}{segments of one-dimensional networks} of length $50$ \KH{on which}{where} stations are distributed following a Poisson point process with \K{node density}{intensity} $\lambda$\KH{}{ ranging from $1$ to $10$}.
The \KHH{network is a unit disk graph with}{} transmission range\KHH{}{ is} $R=1$. 
Ten simulations are run for each combination of $\lambda$ and $\gamma$. 
All simulations last for $2000$ iterations\KH{}{, so when the convergence time is $2000$, it means that additional time is needed for convergence}.
The convergence time and percentage are plotted in Figs.~\ref{fig:time_1d} and~\ref{fig:percentage_1d} respectively.
When $\gamma$ is too small, the effect of annealing is not significant, \khh{and as shown the algorithm may not have converged after $2000$ iterations}{as shown in the figures that after \KH{$2000$ iterations}{a long time}, \KHH{many}{lots of} stations still experience collisions}. 
When $\gamma$ is too large, the convergence time is reduced drastically, but the proportion of stations experiencing collisions is still significant. 
Notice the similarity of the results here with the annealing process in statistical mechanics: when the annealing is too slow, it takes longer time to reach the state with minumum energy; when the annealing is too fast, the system reaches some metastable state or becomes glassy with noncrystalline structure.

\section{Broadcast in Two-Dimensional Networks}
\label{sec:2d}

\subsection{Determining the number of states}
\label{subsec:l_2d}

\begin{figure}[t]
\centering
\subfigure[]{
\includegraphics[width=3.5in]{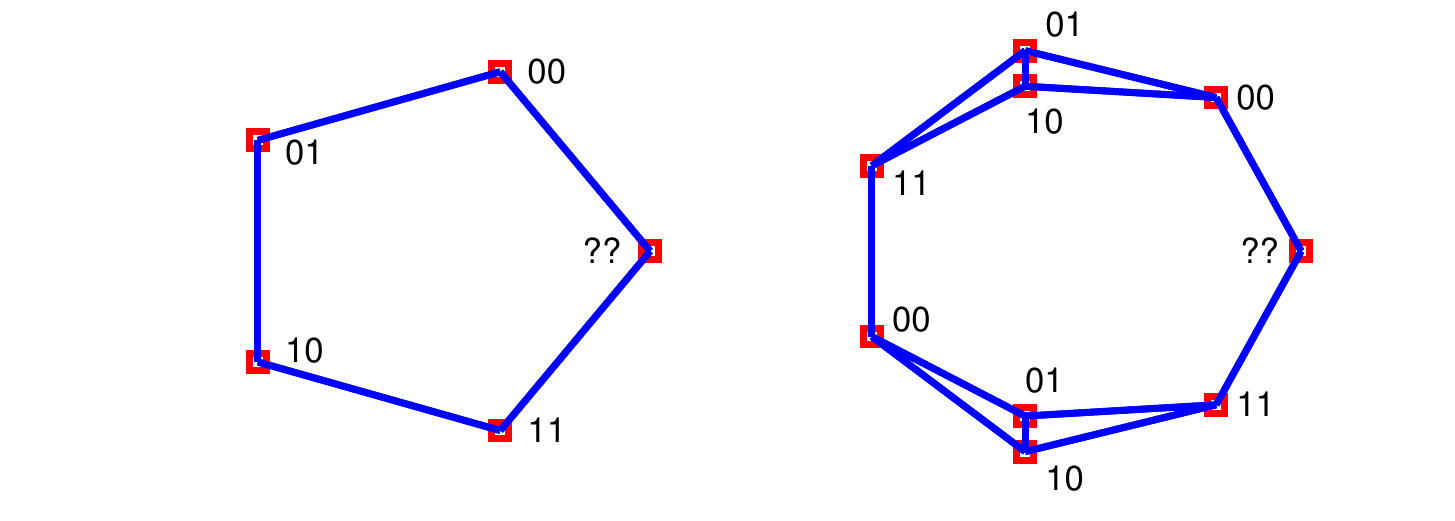}
\label{fig:counterexamples1}
}
\subfigure[]{
\includegraphics[width=3.5in]{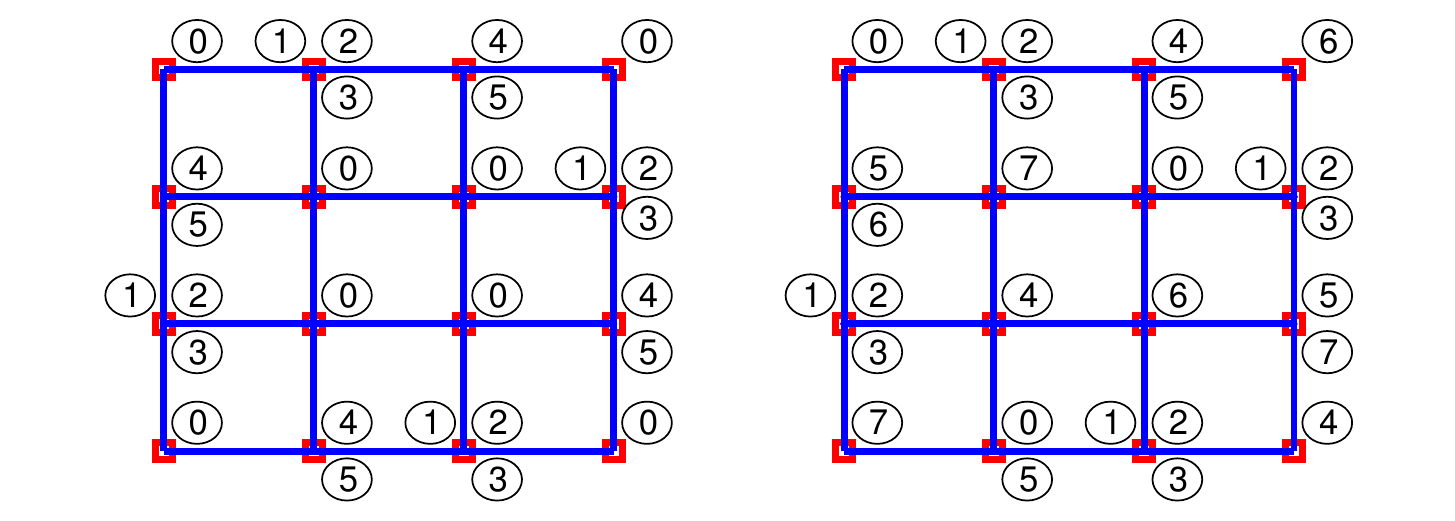}
\label{fig:counterexamples2}
}
\caption[]{Intricacies for two-dimensional networks\KH{: \subref{fig:counterexamples1} determining the number of states (`??' labels stations that are unable to pick a state without collision), \subref{fig:counterexamples2} converging to a collision-free configuration}{}.}
\label{fig:counterexamples}
\end{figure}

Unlike \KH{in}{} one-dimensional networks, the resolution $l_\mathbf{r}$ cannot be completely determined by \khh{(\ref{eqn:resolution_1d})}{two-hop topology information} in two-dimensional networks.
An example is shown in the left part of Fig.~\ref{fig:counterexamples1}. 
If \khh{(\ref{eqn:resolution_1d})}{the \KH{rule}{strategy} in Theorem~\ref{thm:num_states_1d}} is used here, \KH{every station has two one-hop peers and therefore}{all stations} should use \KH{a resolution of}{} four states. 
But since every station is within two hops of every other station,  at least five states are needed to resolve any collision. 
\khh{
For a two-dimensional network, the following theorem shows that (\ref{eqn:resolution_1d}) gives a lower bound on the needed resolution. 
An upper bound on the resolution is also given. 

\begin{thm}
\label{thm:num_states_2d}
\khh{A lower bound on the needed resolution for a collision-free configuration for broadcast to exist is given by}{There exists a collision-free configuration for which the resolution of each station $\mathbf{r}$ is lower bounded by}
\begin{equation}
\label{eqn:resolution_lb_2d}
\underline{l}_\mathbf{r}=\biggl\lceil\log_2\biggl(\max_{\mathbf{r}^\prime\in V_\mathbf{r}\cup\lbrace\mathbf{r}\rbrace}(1+\lvert V_{\mathbf{r}^\prime}\rvert)\biggr)\biggr\rceil.
\end{equation}
\khh{A corresponding upper bound is given by}{and upper bounded by}
\begin{equation}
\label{eqn:resolution_ub_2d}
\overline{l}_\mathbf{r}=\biggl\lceil\log_2\biggl(\max_{\mathbf{r}^\prime\in V_\mathbf{r}^2\cup\lbrace\mathbf{r}\rbrace}(1+\lvert V_{\mathbf{r}^\prime}^2\rvert)\biggr)\biggr\rceil,
\end{equation}
where $V_\mathbf{r}^2$ is the set of one-hop or two-hop peers of $\mathbf{r}$. 
The resulting one-hop broadcast throughput of a two-dimensional network is bounded as follows:
\begin{equation}
\label{eqn:throughput_2d}
\frac{1}{\lvert V\vert}\sum_{\mathbf{r}\in V}\lvert V_\mathbf{r}\rvert2^{-\overline{l}_\mathbf{r}}\le\rho_\text{BC}\le\frac{1}{\lvert V\vert}\sum_{\mathbf{r}\in V}\lvert V_\mathbf{r}\rvert2^{-\underline{l}_\mathbf{r}}.
\end{equation}
\end{thm}

\begin{IEEEproof}
For a station to receive a packet from each one-hop peer, the station itself and all its one-hop peers must transmit at different times. 
If $\mathbf{r}^\prime\in V_\mathbf{r}\cup\lbrace\mathbf{r}\rbrace$, station $\mathbf{r}$ is within the one-hop neighborhood $V_{\mathbf{r}^\prime}\cup\lbrace\mathbf{r}^\prime\rbrace$, and therefore $1+\lvert V_{\mathbf{r}^\prime}\rvert$ states are required to resolve any collision in $V_{\mathbf{r}^\prime}\cup\lbrace\mathbf{r}^\prime\rbrace$. 
Then, at least $\max_{\mathbf{r}^\prime\in V_\mathbf{r}\cup\lbrace\mathbf{r}\rbrace}(1+\lvert V_{\mathbf{r}^\prime}\rvert)$ states are required. 
Finally, we assume $\underline{l}_\mathbf{r}$ to be an integer and $2^{\underline{l}_\mathbf{r}}\ge\max_{\mathbf{r}^\prime\in V_\mathbf{r}\cup\lbrace\mathbf{r}\rbrace}(1+\lvert V_{\mathbf{r}^\prime}\rvert)$, therefore the lower bound (\ref{eqn:resolution_lb_2d}) is established. 

Observe that a station cannot transmit when one of its one-hop or two-hop peers transmits in a collision-free configuration. 
In the worst case, at most one station in every two-hop neighborhood transmits at any time.  
Now, if $\mathbf{r}^\prime\in V_\mathbf{r}^2\cup\lbrace\mathbf{r}\rbrace$, station $\mathbf{r}$ is within the two-hop neighborhood $V_{\mathbf{r}^\prime}^2\cup\lbrace\mathbf{r}^\prime\rbrace$, and in the worst case $1+\lvert V_{\mathbf{r}^\prime}^2\rvert$ states are required to resolve any collision in $V_{\mathbf{r}^\prime}^2\cup\lbrace\mathbf{r}^\prime\rbrace$. 
Therefore, at most $\max_{\mathbf{r}^\prime\in V_\mathbf{r}^2\cup\lbrace\mathbf{r}\rbrace}(1+\lvert V_{\mathbf{r}^\prime}^2\rvert)$ states are required. 
Finally, we assume $\overline{l}_\mathbf{r}$ to be an integer and $2^{\overline{l}_\mathbf{r}}\ge\max_{\mathbf{r}^\prime\in V_\mathbf{r}^2\cup\lbrace\mathbf{r}\rbrace}(1+\lvert V_{\mathbf{r}^\prime}^2\rvert)$, therefore the upper bound (\ref{eqn:resolution_ub_2d}) is established. 
\end{IEEEproof}

The upper bound in Theorem~\ref{thm:num_states_2d} can be quite loose. 
When the network $G$ is `well-connected' it can be improved on by using}{This situation can be remedied if the network $G$ is `well-connected' and the protocol is modified as follows: $l_\mathbf{r}$ is determined by the size of} the maximum clique in $G^2$ containing $\mathbf{r}$, where $G^2=(N,A^2)$ is the square of $G$, \textit{i.e.}, $(\mathbf{r}_i,\mathbf{r}_j)\in A^2$ if $\mathbf{r}_i$ and $\mathbf{r}_j$ are one-hop or two-hop peers in $G$.
\khh{
More formally, we require that $G^2$ to be \textit{chordal}, meaning that in any cycle of at least four vertices, there must exist an edge between some pair of nonadjacent vertices. 
A key property of chordal graphs is that they have a \textit{perfect elimination ordering} of their vertices \cite{DFOG:1}, \textit{i.e.}, one can order the vertices by repeatedly finding a vertex such that all its neighbors form a clique, and then removing it along with all incident edges. 
We use this property to prove the next theorem, which shows that the choice of $l_\mathbf{r}$ based on the maximum clique size in $G^2$ is adequate. 
}{}

\K{}{
\begin{thm}
\label{thm:num_states_2d_chordal}
Suppose a two-dimensional network $G$ has a chordal square, \textit{i.e.}, $G^2$ is chordal. 
If station $\mathbf{r}$ uses resolution $l_\mathbf{r}=\lceil\log_2\lvert C_\mathbf{r}\rvert\rceil$, where $C_\mathbf{r}$ is the maximum clique in $G^2$ containing $\mathbf{r}$, then it is possible for each station to choose a state such that collision-free configurations exist.
\end{thm}
}

\khh{}{
\KH{A graph is \textit{chordal} if in any cycle of at least four vertices, there must exist an edge between some \KHH{pair of}{two} nonadjacent vertices}{A graph is \textit{chordal} if for any cycle of at least four vertices, there must exist an edge between two nonadjacent vertices in the cycle}.
Furthermore, a graph is chordal if and only if it has a perfect elimination ordering of vertices \cite{DFOG:1}. 
A \textit{perfect elimination ordering} \KH{is}{can be} constructed by repeatedly finding a vertex such that all its neighbors form a clique, and then removing it along with all \KH{incident}{} edges\KH{}{ incident on the vertex}. 
The order that the vertices are removed is a perfect elimination ordering. 
\KH{
Evidently, chords are abundant in networks with moderate density. 
On the other hand, if a network is quite sparse, then its square may not be chordal, but \K{medium access control}{scheduling} in a sparse network could be relatively easy through other heuristic means. }{}
}

\K{
\begin{thm}
\label{thm:num_states_2d_chordal}
Suppose a two-dimensional network $G$ has a chordal square, \textit{i.e.}, $G^2$ is chordal. 
If station $\mathbf{r}$ uses resolution $l_\mathbf{r}=\lceil\log_2\lvert C_\mathbf{r}\rvert\rceil$, where $C_\mathbf{r}$ is the maximum clique in $G^2$ containing $\mathbf{r}$, then it is possible for each station to choose a state such that collision-free configurations exist.
\end{thm}
}{}

\begin{IEEEproof}
By assumption, there exists a perfect elimination ordering of vertices for $G^2$. 
Without loss of generality, assume the stations are indexed following the reverse of the perfect elimination ordering, \textit{i.e.}, $\mathbf{r}_j$ appears after $\mathbf{r}_i$ in the perfect elimination ordering if and only if $j<i$. 
We will show by induction that station $\mathbf{r}_i$ must be able to find a state so that it is collision-free with stations $\mathbf{r}_j$ where $j<i$. 

\khh{S}{For s}tation $\mathbf{r}_0$\khh{}{, it} can pick any state. 
Now, assume stations $\mathbf{r}_0,\dots,\mathbf{r}_{i-1}$ pick their states such that they are collision-free among themselves. 
Then, for station $\mathbf{r}_i$, let $C=\lbrace\mathbf{r}_j\colon j<i\text{ and }(\mathbf{r}_i,\mathbf{r}_j)\in A^2\rbrace$. 
By definition of perfect elimination ordering, $C\cup\lbrace\mathbf{r}_i\rbrace$ is a clique in $G^2$. 
Therefore, $\mathbf{r}_i$ and all $\mathbf{r}_j\in C$ use \KHH{resolutions of}{} at least $2^{\lceil\log_2(1+\lvert C\rvert)\rceil}$ states.
Hence, from station $\mathbf{r}_i$'s point of view, there are $\lvert C\rvert$ distinct busy periods, each of length at most $2^{-\lceil\log_2(1+\lvert C\rvert)\rceil}$. 
This means that there is at least one idle slot according to $\mathbf{r}_i$'s resolution, \textit{i.e.}, none of the $\mathbf{r}_j$'s in $C$ use that slot. 
Therefore, station $\mathbf{r}_i$ can pick this slot (or a fraction of this slot if it uses a finer resolution) and therefore becomes collision-free with all stations $\mathbf{r}_j$ where $j<i$. 
Repeating this argument, when station $\mathbf{r}_{\lvert V\rvert-1}$ finds a state that is collision-free with stations $\mathbf{r}_j$ where $j<\lvert V\rvert-1$, then the configuration is now collision-free. 
\end{IEEEproof}

The condition in Theorem~\ref{thm:num_states_2d_chordal} is sufficient but not necessary.
For example, the right part of Fig.~\ref{fig:counterexamples1} shows that a collision-free configuration cannot be found using the \KH{resolutions}{number of states} predicted in Theorem~\ref{thm:num_states_2d}\KH{}{, but the right part of Fig.~\ref{fig:counterexamples2} shows the contrary}. 
\KH{Consider also the right part of Fig.~\ref{fig:counterexamples2}, which is a $4\times4$ square lattice where multiple stations are collocated on some lattice points.
Theorem~\ref{thm:num_states_2d} predicts that every station uses a resolution of eight states, and a collision-free configuration exists, as shown in the figure.
For illustrative purposes, we use \khh{decimal}{octal} representation of the states, \textit{e.g.}, state $101$ is denoted as $5$.}{}
In both cases, $G^2$ is not chordal. 

Since \KH{the sizes of all state spaces}{all number of states} are powers of $2$, additional states are provisioned in many cases. 
Therefore, a collision-free configuration \KHH{is \khh{likely}{guaranteed} to exist by using the rule in Theorem~\ref{thm:num_states_2d_chordal}}{can be formed using local information exchange} even for many networks without chordal squares. 

\khh{}{
For a general two-dimensional network, \KH{upper and lower bounds on the resolution}{Theorem~\ref{thm:num_states_1d} only gives a \textit{lower bound} $\underline{l}_\mathbf{r}\le l_\mathbf{r}$ on the resolution. 
We can obtain a similar expression for an \textit{upper bound} $\overline{l}_\mathbf{r}\ge l_\mathbf{r}$ on the resolution. 
These bounds} are characterized as follows. 

\KH{
\begin{thm}
\label{thm:num_states_2d}
\kh{There exists a self-stabilizing protocol for which the resolution of each station $\mathbf{r}$ is lower bounded by}{\K{A}{The} lower bound\K{}{ $\underline{l}_\mathbf{r}$} on the resolution used by station $\mathbf{r}$ is}\K{
\begin{equation}
\label{eqn:resolution_lb_2d}
\underline{l}_\mathbf{r}=\biggl\lceil\log_2\biggl(\max_{\mathbf{r}^\prime\in V_\mathbf{r}\cup\lbrace\mathbf{r}\rbrace}(1+\lvert V_{\mathbf{r}^\prime}\rvert)\biggr)\biggr\rceil
\end{equation}
}{ computed as follows:
\begin{enumerate}
\item
$\underline{w}_\mathbf{r}=1+\lvert V_\mathbf{r}\rvert$,
\item
$\underline{l}_\mathbf{r}=\lceil\log_2(\max_{\mathbf{r}^\prime\in V_\mathbf{r}\cup\lbrace\mathbf{r}\rbrace}\underline{w}_{\mathbf{r}^\prime})\rceil$.
\end{enumerate}}\kh{and upper bounded by}{\K{A}{The} \KHH{corresponding}{} upper bound\K{}{ $\overline{l}_\mathbf{r}$}\KHH{}{ on the resolution used by station $\mathbf{r}$} is}\K{
\begin{equation}
\label{eqn:resolution_ub_2d}
\overline{l}_\mathbf{r}=\biggl\lceil\log_2\biggl(\max_{\mathbf{r}^\prime\in V_\mathbf{r}^2\cup\lbrace\mathbf{r}\rbrace}(1+\lvert V_{\mathbf{r}^\prime}^2\rvert)\biggr)\biggr\rceil,
\end{equation}
}{ computed as follows:
\begin{enumerate}
\item
$\overline{w}_\mathbf{r}=1+\lvert V_\mathbf{r}^2\rvert$,
\item
$\overline{l}_\mathbf{r}=\lceil\log_2(\max_{\mathbf{r}^\prime\in V_\mathbf{r}^2\cup\lbrace\mathbf{r}\rbrace}\overline{w}_{\mathbf{r}^\prime})\rceil$.
\end{enumerate}}where $V_\mathbf{r}^2$ is the set of one-hop or two-hop peers of $\mathbf{r}$. 
The resulting one-hop broadcast throughput of a two-dimensional network is bounded as follows:
\KHH{\begin{equation}
\label{eqn:throughput_2d}
\frac{1}{\lvert V\vert}\sum_{\mathbf{r}\in V}\lvert V_\mathbf{r}\rvert2^{-\overline{l}_\mathbf{r}}\le\rho_\text{BC}\le\frac{1}{\lvert V\vert}\sum_{\mathbf{r}\in V}\lvert V_\mathbf{r}\rvert2^{-\underline{l}_\mathbf{r}}.
\end{equation}}{\begin{equation}
\label{eqn:throughput_2d}
\bigl\langle\lvert N_\mathbf{r}\rvert2^{-\overline{l}_\mathbf{r}}\bigr\rangle_\mathbf{r}\le\rho_\text{BC}\le\bigl\langle\lvert N_\mathbf{r}\rvert2^{-\underline{l}_\mathbf{r}}\bigr\rangle_\mathbf{r}.
\end{equation}}
\end{thm}

\begin{IEEEproof}
For a station to receive a packet from each one-hop peer, the station itself and all its one-hop peers must transmit at different times. 
If $\mathbf{r}^\prime\in V_\mathbf{r}\cup\lbrace\mathbf{r}\rbrace$, station $\mathbf{r}$ is within the one-hop neighborhood $V_{\mathbf{r}^\prime}\cup\lbrace\mathbf{r}^\prime\rbrace$, and therefore \K{$1+\lvert V_{\mathbf{r}^\prime}\rvert$}{$\underline{w}_{\mathbf{r}^\prime}$} states are required to resolve any collision in $V_{\mathbf{r}^\prime}\cup\lbrace\mathbf{r}^\prime\rbrace$. 
Then, at least \K{$\max_{\mathbf{r}^\prime\in V_\mathbf{r}\cup\lbrace\mathbf{r}\rbrace}(1+\lvert V_{\mathbf{r}^\prime}\rvert)$}{$\max_{\mathbf{r}^\prime\in V_\mathbf{r}\cup\lbrace\mathbf{r}\rbrace}\underline{w}_{\mathbf{r}^\prime}$} states are required. 
Finally, we assume $\underline{l}_\mathbf{r}$ to be an integer and \K{$2^{\underline{l}_\mathbf{r}}\ge\max_{\mathbf{r}^\prime\in V_\mathbf{r}\cup\lbrace\mathbf{r}\rbrace}(1+\lvert V_{\mathbf{r}^\prime}\rvert)$}{$2^{\underline{l}_\mathbf{r}}\ge\max_{\mathbf{r}^\prime\in V_\mathbf{r}\cup\lbrace\mathbf{r}\rbrace}\underline{w}_{\mathbf{r}^\prime}$}, therefore \K{the lower bound (\ref{eqn:resolution_lb_2d}) is established}{we get \KHH{the lower bound}{} $\underline{l}_\mathbf{r}=\lceil\log_2(\max_{\mathbf{r}^\prime\in V_\mathbf{r}\cup\lbrace\mathbf{r}\rbrace}\underline{w}_{\mathbf{r}^\prime})\rceil$}. 

Observe that a station cannot transmit when one of its one-hop or two-hop peers \KHH{transmits}{is active} in a collision-free configuration. 
In the worst case, at most one station in every two-hop neighborhood transmits at any time.  
Now, if $\mathbf{r}^\prime\in V_\mathbf{r}^2\cup\lbrace\mathbf{r}\rbrace$, station $\mathbf{r}$ is within the two-hop neighborhood $V_{\mathbf{r}^\prime}^2\cup\lbrace\mathbf{r}^\prime\rbrace$, and in the worst case \K{$1+\lvert V_{\mathbf{r}^\prime}^2\rvert$}{$\overline{w}_{\mathbf{r}^\prime}$} states are required to resolve any collision in $V_{\mathbf{r}^\prime}^2\cup\lbrace\mathbf{r}^\prime\rbrace$. 
Therefore, at most \K{$\max_{\mathbf{r}^\prime\in V_\mathbf{r}^2\cup\lbrace\mathbf{r}\rbrace}(1+\lvert V_{\mathbf{r}^\prime}^2\rvert)$}{$\max_{\mathbf{r}^\prime\in V_\mathbf{r}^2\cup\lbrace\mathbf{r}\rbrace}\overline{w}_{\mathbf{r}^\prime}$} states are required. 
Finally, we assume $\overline{l}_\mathbf{r}$ to be an integer and \K{$2^{\overline{l}_\mathbf{r}}\ge\max_{\mathbf{r}^\prime\in V_\mathbf{r}^2\cup\lbrace\mathbf{r}\rbrace}(1+\lvert V_{\mathbf{r}^\prime}^2\rvert)$}{$2^{\overline{l}_\mathbf{r}}\ge\max_{\mathbf{r}^\prime\in V_\mathbf{r}^2\cup\lbrace\mathbf{r}\rbrace}\overline{w}_{\mathbf{r}^\prime}$}, therefore \K{the upper bound (\ref{eqn:resolution_ub_2d}) is established}{we get \KHH{the upper bound}{} $\overline{l}_\mathbf{r}=\lceil\log_2(\max_{\mathbf{r}^\prime\in V_\mathbf{r}^2\cup\lbrace\mathbf{r}\rbrace}\overline{w}_{\mathbf{r}^\prime})\rceil$}. 
\end{IEEEproof}
}{
\begin{thm}
\label{thm:num_states_2d_lb}
The lower bound $\underline{l}_\mathbf{r}$ on the resolution used by station $\mathbf{r}$ is computed as follows:
\begin{enumerate}
\item
$\underline{w}_\mathbf{r}=1+\lvert N_\mathbf{r}\rvert$,
\item
 $\underline{l}_\mathbf{r}=\lceil\log_2(\max_{\mathbf{r}^\prime\in N_\mathbf{r}\cup\lbrace\mathbf{r}\rbrace}\underline{w}_{\mathbf{r}^\prime})\rceil$.
\end{enumerate}
\end{thm}

\begin{thm}
\label{thm:num_states_2d_ub}
The upper bound $\overline{l}_\mathbf{r}$ on the resolution used by station $\mathbf{r}$ is computed as follows:
\begin{enumerate}
\item
$\overline{w}_\mathbf{r}=1+\lvert N_\mathbf{r}^2\rvert$,
\item
$\overline{l}_\mathbf{r}=\lceil\log_2(\max_{\mathbf{r}^\prime\in N_\mathbf{r}^2\cup\lbrace\mathbf{r}\rbrace}\overline{w}_{\mathbf{r}^\prime})\rceil$.
\end{enumerate}
Here, $N_\mathbf{r}^2$ is the set of one-hop or two-hop peers of $\mathbf{r}$. 
\end{thm}

\begin{proof}
Observe that a station cannot transmit when one of its one-hop or two-hop peers is active in a collision-free configuration. 
In the worst case, at most one station in every two-hop neighborhood transmits at any time.  
Now, if $\mathbf{r}^\prime\in N_\mathbf{r}^2\cup\lbrace\mathbf{r}\rbrace$, station $\mathbf{r}$ is within the two-hop neighborhood $N_{\mathbf{r}^\prime}^2\cup\lbrace\mathbf{r}^\prime\rbrace$, and in the worst case $\overline{w}_{\mathbf{r}^\prime}$ states are required to resolve any collision in $N_{\mathbf{r}^\prime}^2\cup\lbrace\mathbf{r}^\prime\rbrace$. 
Therefore, at most $\max_{\mathbf{r}^\prime\in N_\mathbf{r}^2\cup\lbrace\mathbf{r}\rbrace}\overline{w}_{\mathbf{r}^\prime}$ states are required. 
Finally, we assume $\overline{l}_\mathbf{r}$ to be an integer and $2^{\overline{l}_\mathbf{r}}\ge\max_{\mathbf{r}^\prime\in N_\mathbf{r}^2\cup\lbrace\mathbf{r}\rbrace}\overline{w}_{\mathbf{r}^\prime}$, therefore we get $\overline{l}_\mathbf{r}=\lceil\log_2(\max_{\mathbf{r}^\prime\in N_\mathbf{r}^2\cup\lbrace\mathbf{r}\rbrace}\overline{w}_{\mathbf{r}^\prime})\rceil$. 
\end{proof}
}
}

\subsection{\kh{Multi-resolution MAC Protocol for Two-dimensional Networks}{A Modified Multi-Resolution MAC Protocol}}
\label{sec:multi_resolution_2d}

\kh{Protocol~\ref{alg:mr_bc} with $\epsilon=0$ does not always}{The multi-resolution protocol\KH{}{ we propose} for one-dimensional networks may not} work for all two-dimensional networks. 
\khh{In particular, the resulting Markov chain can have an absorbing class with more than one configuration, none of which is collision-free.}{}
An example is illustrated in Fig.~\ref{fig:counterexamples2}.\KH{}{This is a $4\times4$ square lattice, where multiple stations are collocated on some lattice points. }
\KH{Every station uses a resolution of eight states.}{}
\KH{T}{Every station uses eight states, and t}he right part of Fig.~\ref{fig:counterexamples2} shows that a collision-free configuration exists. 
But, if the initial configuration is the one shown in the left part of Fig.~\ref{fig:counterexamples2}, and the protocol\KH{}{ proposed} in Section~\ref{sec:1d} is used\KH{, then the following occurs}{ here, then}: 
\begin{enumerate}
\item
All stations in initial states \KH{$1,2,3,4,5$}{$001,010,011,100,101$} remain in their current states with probability one, since they do not cause any collision. 
All stations in initial state \KH{$0$}{$000$} can only choose \KH{$0,6,7$}{$000,110,111$} as their next states, because all other states are not available. 
This repeats for all subsequent iterations. 
\item
Consider the four stations in the middle of the network, which\KH{}{ all} have initial states \KH{$0$}{$000$}. 
\KH{They}{These stations} are within two hops of each other, so they must use different states. 
However, only states \KH{$0,6,7$}{$000,110,111$} are available to them in any cycle. 
\KH{Hence}{Therefore}, collision-free configurations cannot be reached. 
\end{enumerate}

\kh{Choosing $\epsilon>0$ in Protocol~\ref{alg:mr_bc}}{A simple modification} prevents the preceding deadlock. 
For any station, if the votes received do not all point to a single state, then the station increases the total weight of the votes received \KH{for}{by} any state by a \KH{nonzero}{small} constant. 
\textit{A station in this situation will have nonozero probability of choosing any state to be its next state}. 
\kh{}{\KHH{The modified multi-resolution protocol is shown as Protocol~\ref{alg:mr_bc}\K{, with $\epsilon>0$}{}.}{}}This randomization does not affect any absorption configuration, and is necessary to establish the counterpart of Theorem~\ref{thm:mr_1d} for two-dimensional networks. 

\K{}{
\KHH{
\begin{algorithm}[t]
\caption{Multi-Resolution MAC Protocol for Broadcast on Two-Dimensional Networks}
\label{alg:mr_bc_2d}
\begin{algorithmic}[1]
\STATE
$\mathbf{r}$ sets the votes on all states to zero. 
\FOR{$\mathbf{r}^\prime\in V_\mathbf{r}\cup\lbrace\mathbf{r}\rbrace$}
\IF{$\mathbf{r}$ is the only station occupying its current state in station $\mathbf{r}^\prime$'s one-hop neighborhood} 
\STATE
\KH{$\mathbf{r}$'s current state is assigned a single vote of weight one}{A single vote of weight one on $\mathbf{r}$'s current state is given by $\mathbf{r}^\prime$}. 
\ELSE
\STATE
\KH{$\mathbf{r}$ determines which states (according to $\mathbf{r}$'s resolution) are idle or have collisions in $\mathbf{r}^\prime$'s one-hop neighborhood}{$\mathbf{r}$ determines the states (according to $\mathbf{r}$'s resolution) that station $\mathbf{r}^\prime$ is idle or collides}. 
\STATE
A vote of weight $\frac{1}{n}$ is given to each such state\KH{}{ by station $\mathbf{r}^\prime$}, where $n$ is the number of such states. 
\ENDIF
\ENDFOR
\IF{$n_s>0$ for multiple $s$'s, where $n_s$ is the total weight state $s$ receives}
\STATE
Replace $n_s$ by $n_s+\epsilon$, where $\epsilon>0$, for all $s$. 
\ENDIF
\STATE
$\mathbf{r}$ selects state $s$ with a probability proportional to $f(n_s)$.
\end{algorithmic}
\end{algorithm}
}{}
}

\begin{thm}
\label{thm:mr_2d}
For a two-dimensional network, suppose \KH{each station uses}{all stations use} a sufficiently \KH{fine resolution (\textit{e.g.}, the upper bound in Theorem~\ref{thm:num_states_2d}),}{large number of states} so that the existence of collision-free configurations is guaranteed. 
Then, starting from an arbitrary initial configuration, \KHH{Protocol~\ref{alg:mr_bc} \kh{with $\epsilon>0$}{}}{the modified multi-resolution protocol} will \khh{converge to}{, after a sufficiently long time, \KH{result in}{lead to}} a collision-free configuration. 
\end{thm}

\begin{IEEEproof}
As in the proof of Theorem~\ref{thm:mr_1d}, every collision-free configuration is absorbing. 
So we only need to consider the case when the initial configuration is not collision-free. 

The remaining proof is similar to the \khh{proof of Theorem 1}{one} in \cite{DLPC:1}\KHH{.}{, \textit{i.e.},} 
\KHH{W}{w}e will  show that an all-zero configuration\KHH{, \textit{i.e.}, a configuration where every station's state is a binary string of all zeros,}{} is reached with nonzero probability, and then in the next cycle there is a nonzero probability of reaching a collision-free configuration. 
Without loss of generality, assume station $\mathbf{r}_0$ collides with some station. 
Consider a spanning tree rooted at $\mathbf{r}_0$, and assume the stations are indexed following the breadth-first search order. 
\KHH{Using Protocol~\ref{alg:mr_bc} \kh{with $\epsilon>0$}{}, t}{T}he following happens with nonzero probability: 
\begin{itemize}
\item
Station $\mathbf{r}_0$ chooses a state \khh{that}{to} collide\khh{s}{} with its child\khh{}{ren} with the smallest index in the spanning tree, and repeats this for all children in the spanning tree following the breadth-first search ordering in subsequent cycles. 
After colliding with all children, it then chooses the all-zero state and remains in that state. 
\item
For station $\mathbf{r}_i$: 
\begin{enumerate}
\item
If it does not collide with its parent in the spanning tree, then it remains in its current state \KHH{until it collides with its parent}{}. 
\item
\KHH{When it collides with its parent, it follows what station $\mathbf{r}_0$ does, \textit{i.e.}}{Otherwise, like $\mathbf{r}_0$}, it chooses a state \khh{that}{to} collide\khh{s}{} with its child\khh{}{ren} with the smallest index in the spanning tree, and repeats this for all children in the spanning tree following the breadth-first search ordering in subsequent cycles. 
After colliding with all children, it then chooses the all-zero state and remains in that state. 
\end{enumerate}
\end{itemize}
Finally, all stations are in the all-zero state, \textit{i.e.}, every station \KHH{collides}{is colliding} with all one-hop and two-hop peers. 
Therefore, in the next cycle, there is a nonzero probability that \khh{the stations choose any configuration including one that is collision-free}{every station chooses a state such that the configuration is collision-free}, where the existence of collision-free configurations is guaranteed. 
Hence, all configurations with collision\khh{s}{} are transient. 
\end{IEEEproof}

\subsection{Simulations: Dynamically Adjusting the Number of States}
\label{subsec:annealing_2d}

For a general two-dimensional network, it is difficult for stations to predict the \KH{resolutions}{number of states} they need. 
It may still be difficult even for networks with chordal squares, since it is not known whether it is possible to find the maximum clique in the square of a\KH{}{ unit disk} graph efficiently. 
Therefore, we propose the following \KH{dynamic algorithm}{dynamics to solve the problem}. 
Initially, every station \KH{sets its resolution to be the lower bound given by Theorem~\ref{thm:num_states_2d}}{follows the rule proposed for one-dimensional networks, \textit{i.e.}, station $\mathbf{r}$ uses resolution $l_\mathbf{r}=\lceil\log_2(\max_{\mathbf{r}^\prime\in N_\mathbf{r}\cup\lbrace\mathbf{r}\rbrace}w_{\mathbf{r}^\prime})\rceil$, where $w_\mathbf{r}=1+\lvert N_\mathbf{r}\rvert$,} and executes the modified multi-resolution protocol with annealing. 
\KH{If a station knows that the local configuration within its two-hop neighborhood remains the same for a number of iterations ($10$ in our simulations), but it still experiences collisions, then it checks if there are any idle states within its two-hop neighborhood. 
If such states exist, it selects one of these states; otherwise, it doubles the size of its state space (\textit{i.e.}, it `refines' its resolution), picks its state randomly, resets the strength of interaction it uses and continues executing the protocol.
The refinement stops once the local configuration is collision-free, or the upper bound given by Theorem~\ref{thm:num_states_2d} is reached, whichever occurs first.
The upper bound provides a guarantee on the minimum rate a station can have. }{}

\begin{figure}[t]
\centering
\subfigure[Convergence time]{
\includegraphics[width=3.5in]{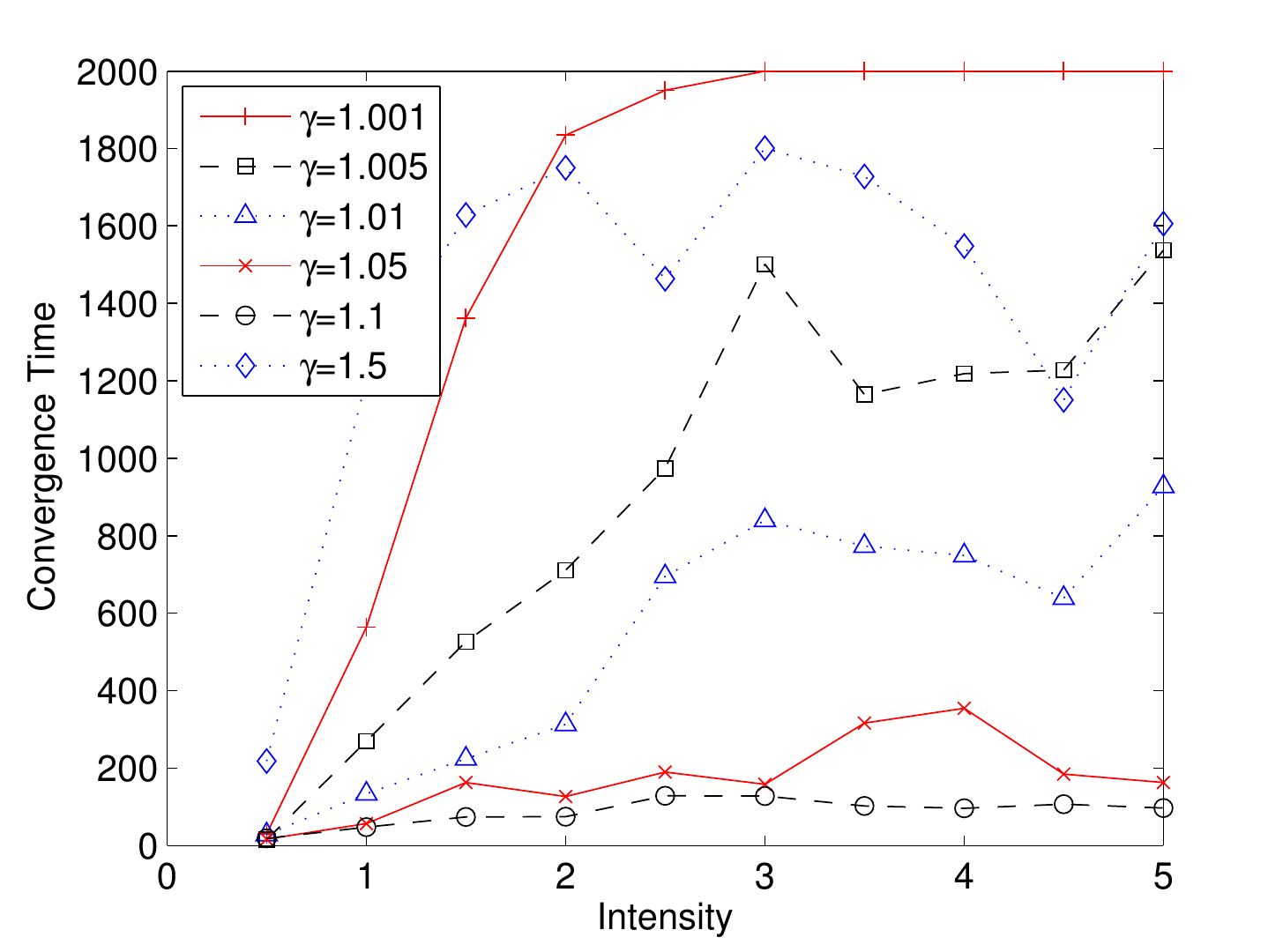}
\label{fig:time_2d}
}
\hspace{-0.4in}
\subfigure[Convergence percentage]{
\includegraphics[width=3.5in]{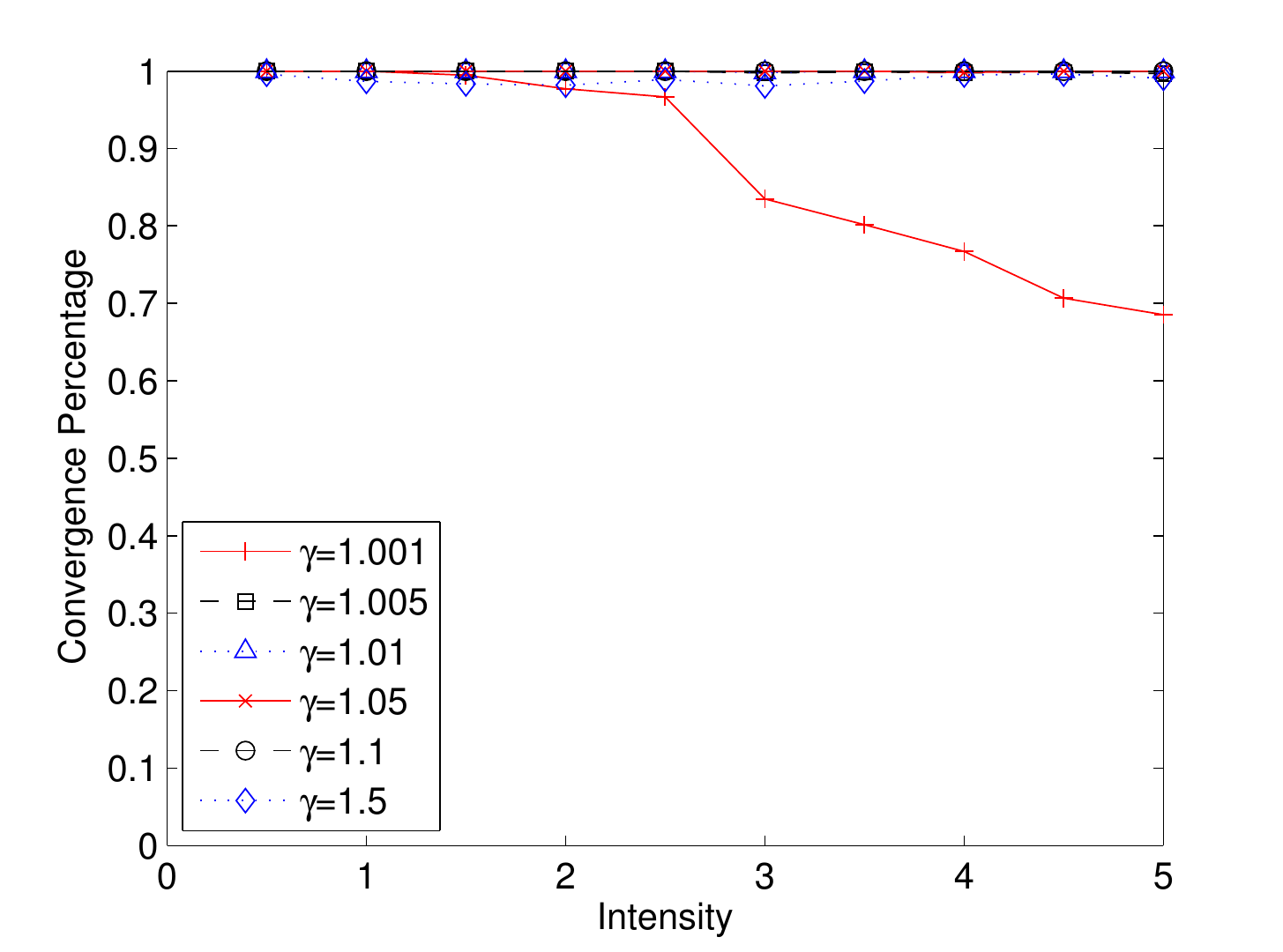}
\label{fig:percentage_2d}
}
\caption[]{Simulations of the multi-resolution protocol with annealing for two-dimensional networks.}
\label{fig:annealing_2d}
\end{figure}

For simulation, we consider \khh{a}{} $10\times10$ square \KH{area}{segments} of two-dimensional networks where stations are distributed following a Poisson point process with \K{node density}{intensity} $\lambda$\KH{}{ ranging from $0.5$ to $5$}. 
All other simulation settings are the same as those for one-dimensional networks. 
The convergence time and percentage are plotted in Figs.~\ref{fig:time_2d} and~\ref{fig:percentage_2d}\khh{,}{} respectively. 
\KH{T}{Notice t}he convergence time is longer compared to one-dimensional networks, since\KH{}{ the} stations may need to adjust \KH{their}{the} resolutions\KH{}{ they use}. 
\KH{W}{Also, w}hen $\gamma$ is too large, the convergence time increases drastically. 
\KHH{In this case}{For large $\gamma$}, the interaction between stations is so large that the protocol behaves like \textit{majority vote} shortly after the protocol is executed. 
This makes the randomization of states after each refinement not effective in resolving collision\KH{s}{}. 

\section{Extension to Multicast}
\label{sec:mc}

\begin{figure}[t]
\centering
\subfigure[]{
\includegraphics[width=3.3in]{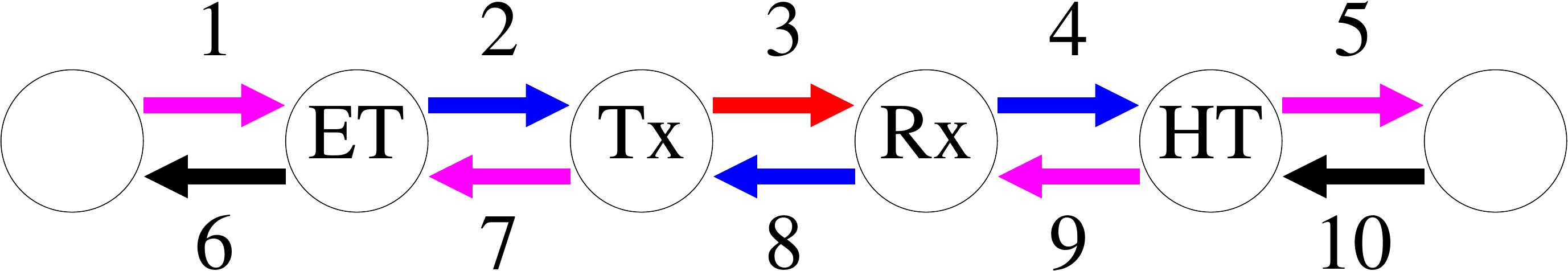}
\label{fig:multicast}
}
\subfigure[]{
\includegraphics[width=2.7in]{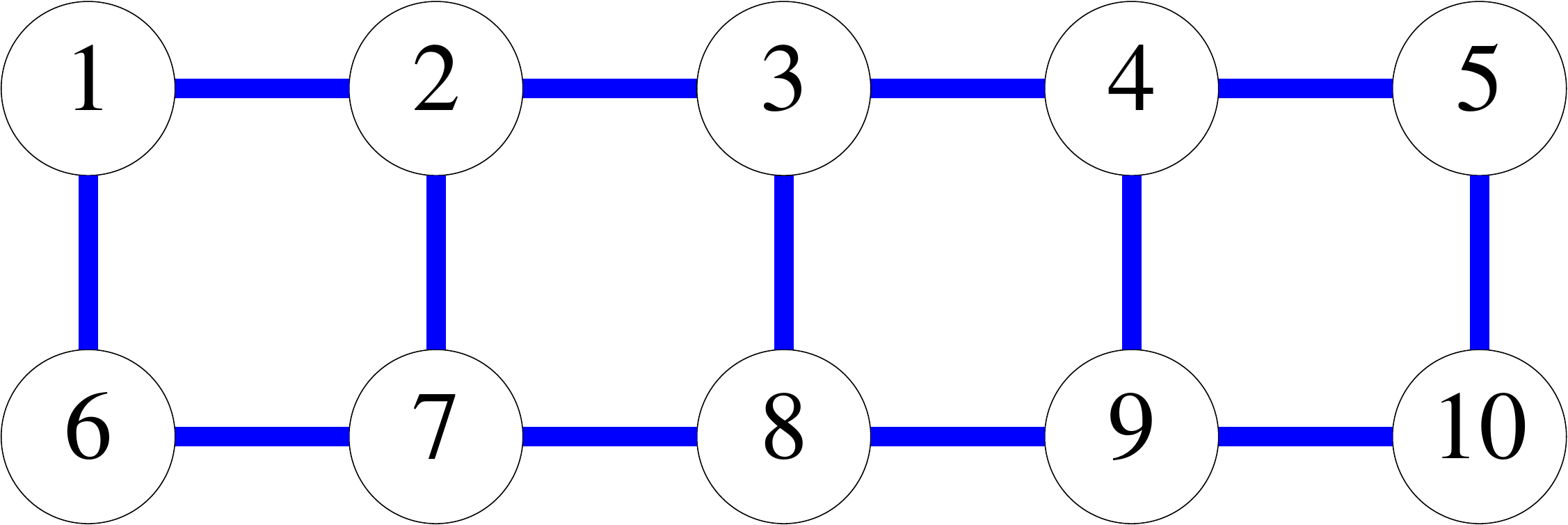}
\label{fig:auxiliary}
}
\caption[]{\KHH{\subref{fig:multicast}}{} Multicast on a graph \KHH{$G$}{} is equivalent to \KHH{\subref{fig:auxiliary}}{} broadcast on the corresponding auxiliary graph \KHH{$G^L$}{}.}
\label{fig:mc_bc}
\end{figure}

In this section we extend the results\khh{}{ on broadcast} to multicast \khh{traffic}{}. 
\KHH{Here\khh{,}{} we consider every undirected link in the network $G$ to be two \textit{directed} links in opposite directions.}{}
We define two directed links $a,a^\prime$ in $G$ to be (one-hop) \textit{link-peers} if and only if the transmitter of one link is the receiver of the other link. 
The neighboring relationship can be represented by an \textit{auxiliary graph} $G^L$ constructed as follows: a directed link $a$ in $G$ is represented by a vertex $a$ in $G^L$, and two directed links $a,a^\prime$ being one-hop link-peers in $G$ is represented by an undirected edge $(a,a^\prime)$ in $G^L$. 
An example is shown in Fig.~\ref{fig:mc_bc}. 
Suppose station Tx transmits packets to station Rx only, \textit{i.e.}, only link $3$\KHH{}{ (labeled in red)} is used by Tx, and suppose it is active. 
Then all its\KHH{}{ (blue)} one-hop link-peers \KHH{(links $2,4,8$)}{} must be silent due to the half-duplex constraint. 
All\KHH{}{ (magenta)} two-hop link-peers must be silent also, because either they are originated from \textit{hidden terminals} (links $1,5,9$), or they are links not used by Tx (link $7$). 
All other links\KHH{}{ (black)} are free to transmit, because either they are links from \textit{exposed terminals} to other stations not interfered by Tx (link $6$), or they are sufficiently far away (link $10$). 
Therefore, under this neighboring model, a link must choose a state different from any of its one-hop and two-hop link-peers in order to form a collision-free configuration. 
Correspondingly, in the auxiliary graph $G^L$, when vertex $3$ is active, its one-hop and two-hop peers must be silent at the same time. 
In general, \textit{if we associate a state variable to each link, which always takes the same value as the state variable of the transmitter of the link (meaning that all links used by the same multicast session must be in the same state), then multicast on a graph $G$ is equivalent to broadcast on the corresponding \textit{auxiliary graph} $G^L$}. 
Therefore, most results obtained for broadcast in previous sections can be applied here with slight modifications. 

As discussed in Section~\ref{sec:model}, we assume station $\mathbf{r}$ multicasts packets to a subset $D_\mathbf{r}$ of its one-hop peers. 
We represent a multicast session by the set of links used, \textit{i.e.}, $M_\mathbf{r}=\lbrace(\mathbf{r},\mathbf{r}^\prime)\colon\mathbf{r}^\prime\in D_\mathbf{r}\rbrace$. 
\KHH{Following the ideas in Section~\ref{sec:2d}, we use one-hop and two-hop \textit{link-neighborhood} size to estimate the lower and upper bounds on the resolution for multicast.}{Recall that we use the one-hop and two-hop neighborhood size to estimate the lower and upper bounds on the resolution when there is only broadcast traffic. 
By considering the analogy between multicast on $G$ and broadcast on $G^L$, it will be natural to use the one-hop and two-hop \textit{link-neighborhood} size to estimate the corresponding lower and upper bounds on the resolution for multicast. }

\KH{
\begin{thm}
\label{thm:num_states_2d_mc}
\khh{A lower bound on the needed resolution for a collision-free configuration for multicast to exist}{\kh{There exists a self-stabilizing protocol for multicast for which the lower bound on the resolution of each station $\mathbf{r}$}{For multicast, \kh{to allow stations to have a fair share of the wireless channel with their respective neighbors while avoiding collisions, a}{\K{a}{the}} lower bound\K{}{ $\underline{l}_\mathbf{r}$} on the resolution used by station $\mathbf{r}$}} is computed as follows:\K{
\begin{IEEEeqnarray}{rCl}
\underline{I}_a&=&\lbrace\mathbf{r}^\prime\colon M_{\mathbf{r}^\prime}\cap(L_a\cup\lbrace a\rbrace)\ne\emptyset\rbrace,\\
\underline{l}_a&=&\max_{a^\prime\in L_a\cup\lbrace a\rbrace}\lvert\underline{I}_{a^\prime}\rvert,\\
\label{eqn:resolution_mc_lb}
\underline{l}_\mathbf{r}&=&\biggl\lceil\log_2\biggl(\max_{a\in M_\mathbf{r}}\underline{l}_a\biggr)\biggr\rceil,
\end{IEEEeqnarray}
where\kh{}{ the set $\underline{I}_a$ contains the multicast sessions that use any link within link $a$'s one-hop link-neighborhood, and}}{
\begin{enumerate}
\item
$\underline{I}_a=\lbrace\mathbf{r}^\prime\colon M_{\mathbf{r}^\prime}\cap(L_a\cup\lbrace a\rbrace)\ne\emptyset\rbrace$,
\item
$\underline{w}_a=\lvert\underline{I}_a\rvert$,
\item
$\underline{l}_a=\max_{a^\prime\in L_a\cup\lbrace a\rbrace}\underline{w}_{a^\prime}$,
\item
$\underline{l}_\mathbf{r}=\lceil\log_2(\max_{a\in M_\mathbf{r}}\underline{l}_a)\rceil$.
\end{enumerate}
Here,} $L_a$ is the set of one-hop link-peers of link $a$\kh{, and $\underline{I}_a$ contains the multicast sessions that use any link within link $a$'s one-hop link-neighborhood\khh{. A}{; and the}}{.
\K{A}{The}} \KHH{corresponding}{} upper bound\K{}{ $\overline{l}_\mathbf{r}$}\KHH{}{ on the resolution used by station $\mathbf{r}$} is computed as follows:\K{
\begin{IEEEeqnarray}{rCl}
\overline{I}_\mathbf{r}&=&\lbrace\mathbf{r}^\prime\colon M_{\mathbf{r}^\prime}\cap\cup_{a\in M_\mathbf{r}}(L_a^2\cup\lbrace a\rbrace)\ne\emptyset\rbrace,\\
\label{eqn:resolution_mc_ub}
\overline{l}_\mathbf{r}&=&\biggl\lceil\log_2\biggl(\max_{\mathbf{r}^\prime\in\overline{I}_\mathbf{r}}\lvert\overline{I}_{\mathbf{r}^\prime}\rvert\biggr)\biggr\rceil,
\end{IEEEeqnarray}
where\kh{}{ any multicast session in the set $\overline{I}_\mathbf{r}$ either is $M_\mathbf{r}$ or cannot be active at the same time as $M_\mathbf{r}$ because it uses a link within the two-hop link-neighborhood of a link used by $M_\mathbf{r}$, and}}{
\begin{enumerate}
\item
$\overline{I}_\mathbf{r}=\lbrace\mathbf{r}^\prime\colon M_{\mathbf{r}^\prime}\cap\cup_{a\in M_\mathbf{r}}(L_a^2\cup\lbrace a\rbrace)\ne\emptyset\rbrace$,
\item
$\overline{w}_\mathbf{r}=\lvert\overline{I}_\mathbf{r}\rvert$,
\item
$\overline{l}_\mathbf{r}=\lceil\log_2(\max_{\mathbf{r}^\prime\in\overline{I}_\mathbf{r}}\overline{w}_{\mathbf{r}^\prime})\rceil$.
\end{enumerate}
Here,} $L_a^2$ is the set of one-hop or two-hop link-peers of link $a$\kh{, and $\overline{I}_\mathbf{r}$ contains $M_\mathbf{r}$ and all multicast sessions that cannot be active at the same time as $M_\mathbf{r}$ because they use links within the two-hop link-neighborhood of a link used by $M_\mathbf{r}$}{}.
The resulting one-hop multicast throughput of a two-dimensional network is bounded as follows:
\KHH{\begin{equation}
\label{eqn:throughput_mc_2d}
\frac{1}{\lvert V\rvert}\sum_{\mathbf{r}\in V}\lvert D_\mathbf{r}\rvert2^{-\overline{l}_\mathbf{r}}\le\rho_\text{MC}\le\frac{1}{\lvert V\rvert}\sum_{\mathbf{r}\in V}\lvert D_\mathbf{r}\rvert2^{-\underline{l}_\mathbf{r}}.
\end{equation}}{\begin{equation}
\label{eqn:throughput_mc_2d}
\bigl\langle\lvert D_\mathbf{r}\rvert2^{-\overline{l}_\mathbf{r}}\bigr\rangle_\mathbf{r}\le\rho_\text{MC}\le\bigl\langle\lvert D_\mathbf{r}\rvert2^{-\underline{l}_\mathbf{r}}\bigr\rangle_\mathbf{r}.
\end{equation}}
\end{thm}

\begin{IEEEproof}
\K{}{The set $\underline{I}_a$ contains the multicast sessions that use any link in link $a$'s one-hop link-neighborhood. }
\K{For any link $a$}{At any time}, at most one multicast session in $\underline{I}_a$ can be active \K{at any time}{}. 
Therefore, at least \K{$\lvert\underline{I}_a\rvert$}{$\lvert\underline{I}_a\rvert=\underline{w}_a$} states are required to resolve collisions among the multicast sessions \KHH{in $\underline{I}_a$}{using links in link $a$'s one-hop link-neighborhood}. 
Link $a$ belongs to the one-hop link-neighborhood of any link $a^\prime\in L_a\cup\lbrace a\rbrace$. 
Therefore, if link $a$ is used by any multicast session, it needs at least \K{$\max_{a^\prime\in L_a\cup\lbrace a\rbrace}\lvert\underline{I}_{a^\prime}\rvert=\underline{l}_a$}{$\max_{a^\prime\in L_a\cup\lbrace a\rbrace}\underline{w}_{a^\prime}=\underline{l}_a$} states to resolve collisions. 
Finally, station $\mathbf{r}$ should use the finest resolution that its links use, therefore we have \K{the lower bound (\ref{eqn:resolution_mc_lb})}{$\underline{l}_\mathbf{r}=\lceil\log_2(\max_{a\in M_\mathbf{r}}\underline{l}_a)\rceil$ as the lower bound on the resolution}.

\K{}{Any multicast session in the set $\overline{I}_\mathbf{r}$ either is $M_\mathbf{r}$ or cannot be active at the same time as $M_\mathbf{r}$ because it uses a link in the two-hop link-neighborhood of a link used by $M_\mathbf{r}$. }
\K{For any station $\mathbf{r}$}{In the worst case}, at most one multicast session in $\overline{I}_\mathbf{r}$ can be active at any time \K{in the worst case}{}. 
Therefore, at most \K{$\lvert\overline{I}_\mathbf{r}\rvert$}{$\lvert\overline{I}_\mathbf{r}\rvert=\overline{w}_\mathbf{r}$} states are required to resolve collisions among the multicast sessions in $\overline{I}_\mathbf{r}$. 
Finally, station $\mathbf{r}$ also belongs to $\overline{I}_{\mathbf{r}^\prime}$ for $\mathbf{r}^\prime\in\overline{I}_\mathbf{r}$, implying the upper bound \K{(\ref{eqn:resolution_mc_ub})}{$\lceil\log_2(\max_{\mathbf{r}^\prime\in\overline{I}_\mathbf{r}}\overline{w}_{\mathbf{r}^\prime})\rceil=\overline{l}_\mathbf{r}$ on the resolution}.
\end{IEEEproof}
}{
\begin{thm}
\label{thm:num_states_2d_mc_lb}
For multicast, the lower bound $\underline{l}_\mathbf{r}$ on the resolution used by station $\mathbf{r}$ is computed as follows:
\begin{enumerate}
\item
$\underline{I}_a=\lbrace\mathbf{r}^\prime\colon S_{\mathbf{r}^\prime}\cap(L_a\cup\lbrace a\rbrace)\ne\emptyset\rbrace$,
\item
$\underline{w}_a=\lvert\underline{I}_a\rvert$,
\item
$\underline{l}_a=\max_{a^\prime\in L_a\cup\lbrace a\rbrace}\underline{w}_{a^\prime}$,
\item
$\underline{l}_\mathbf{r}=\lceil\log_2(\max_{a\in S_\mathbf{r}}\underline{l}_a)\rceil$.
\end{enumerate}
Here, $L_a$ is the set of one-hop link-peers of link $a$. 
\end{thm}

\begin{proof}
The set $\underline{I}_a$ contains the multicast sessions that use any link in link $a$'s one-hop link-neighborhood. 
At any time, at most one multicast session in $\underline{I}_a$ can be active. 
Therefore, at least $\lvert\underline{I}_a\rvert=\underline{w}_a$ states are required to resolve collision among the multicast sessions using links in link $a$'s one-hop link-neighborhood. 
Link $a$ belongs to the one-hop link-neighborhood of any link $a^\prime\in L_a\cup\lbrace a\rbrace$. 
Therefore, if link $a$ is used by any multicast session, it needs at least $\max_{a^\prime\in L_a\cup\lbrace a\rbrace}\underline{w}_{a^\prime}=\underline{l}_a$ states to resolve collision. 
Finally, station $\mathbf{r}$ should use the finest resolution that its links use, therefore we have $\underline{l}_\mathbf{r}=\lceil\log_2(\max_{a\in S_\mathbf{r}}\underline{l}_a)\rceil$ as the lower bound on the resolution.
\end{proof}

\begin{thm}
\label{thm:num_states_2d_mc_ub}
For multicast, the upper bound $\overline{l}_\mathbf{r}$ on the resolution used by station $\mathbf{r}$ is computed as follows:
\begin{enumerate}
\item
$\overline{I}_\mathbf{r}=\lbrace\mathbf{r}^\prime\colon S_{\mathbf{r}^\prime}\cap\cup_{a\in S_\mathbf{r}}(L_a^2\cup\lbrace a\rbrace)\ne\emptyset\rbrace$,
\item
$\overline{w}_\mathbf{r}=\lvert\overline{I}_\mathbf{r}\rvert$,
\item
$\overline{l}_\mathbf{r}=\lceil\log_2(\max_{\mathbf{r}^\prime\in\overline{I}_\mathbf{r}}\overline{w}_{\mathbf{r}^\prime})\rceil$.
\end{enumerate}
Here, $L_a^2$ is the set of one-hop or two-hop link-peers of link $a$. 
\end{thm}

\begin{proof}
Any multicast session in the set $\overline{I}_\mathbf{r}$ either is $S_\mathbf{r}$ or cannot be active at the same time as $S_\mathbf{r}$ because it uses a link in the two-hop link-neighborhood of a link used by $S_\mathbf{r}$. 
In the worst case, at most one multicast session in $\overline{I}_\mathbf{r}$ can be active at any time. 
Therefore, at most $\lvert\overline{I}_\mathbf{r}\rvert=\overline{w}_\mathbf{r}$ states are required to resolve collision among the multicast sessions in $\overline{I}_\mathbf{r}$. 
Finally, station $\mathbf{r}$ also belongs to $\overline{I}_{\mathbf{r}^\prime}$ for $\mathbf{r}^\prime\in\overline{I}_\mathbf{r}$, implying the upper bound $\lceil\log_2(\max_{\mathbf{r}^\prime\in\overline{I}_\mathbf{r}}\overline{w}_{\mathbf{r}^\prime})\rceil=\overline{l}_\mathbf{r}$ on the resolution.
\end{proof}
}

Note that \KH{Theorem~\ref{thm:num_states_2d_mc} gives the same result as Theorem~\ref{thm:num_states_2d}}{Theorems~\ref{thm:num_states_2d_mc_lb} and~\ref{thm:num_states_2d_mc_ub} give the same result as Theorems~\ref{thm:num_states_2d_lb} and~\ref{thm:num_states_2d_ub}} when there is only broadcast, since $\underline{I}_a=V_\mathbf{r}\cup\lbrace\mathbf{r}\rbrace$ when $\mathbf{r}$ is the transmitter of link $a$, and $\overline{I}_\mathbf{r}=V_\mathbf{r}^2\cup\lbrace\mathbf{r}\rbrace$.

The corresponding multi-resolution MAC protocol for multicast is shown as Protocol~\ref{alg:mr_mc}. 
\KHH{T}{Apart from the randomization discussed in Section~\ref{sec:2d} (lines $10-12$ \KH{in Protocol~\ref{alg:mr_mc}}{}), t}he only difference from Protocol~\ref{alg:mr_bc} is that \textit{the votes on station $\mathbf{r}$'s next state are cast\khh{}{ed} by using the state information of link $a^\prime$'s one-hop link-neighborhood (lines $3$ and $6$), where $a^\prime$ belongs to any one-hop link-neighborhood of link $a$ used by station $\mathbf{r}$ (line $2$)}.
\KHH{By considering the analogy between multicast on $G$ and broadcast on $G^L$, we have the following convergence result. }{}

\begin{algorithm}[t]
\caption{Multi-Resolution MAC Protocol for Multicast}
\label{alg:mr_mc}
\begin{algorithmic}[1]
\WHILE{\khh{station $\mathbf{r}$ is active}{}}
\STATE
$\mathbf{r}$ sets the votes on all states to zero. 
\FOR{$a^\prime\in\cup_{a\in M_\mathbf{r}}(L_a\cup\lbrace a\rbrace)$}
\IF{$\mathbf{r}$ is the only station occupying its current state in link $a^\prime$'s one-hop link-neighborhood} 
\STATE
\KH{$\mathbf{r}$'s current state is assigned a single vote of weight one}{A single vote of weight one on $\mathbf{r}$'s current state is given by $a^\prime$}. 
\ELSE
\STATE
\KH{$\mathbf{r}$ determines which states (according to $\mathbf{r}$'s resolution) are idle or have collisions in link $a^\prime$'s one-hop link-neighborhood}{$\mathbf{r}$ determines the states (according to $\mathbf{r}$'s resolution) that are idle or experience collision in link $a^\prime$'s neighborhood}. 
\STATE
A vote of weight $\frac{1}{n}$ is \khh{added}{given} to each such state\KH{}{ by link $a^\prime$}, where $n$ is the number of such states. 
\ENDIF
\ENDFOR
\IF{$n_s>0$ for multiple $s$'s, where $n_s$ is the total weight state $s$ receives}
\STATE
Replace $n_s$ by $n_s+\epsilon$, where $\epsilon>0$, for all $s$. 
\ENDIF
\STATE
$\mathbf{r}$ selects state $s$ with a probability proportional to $f(n_s)$.
\ENDWHILE
\end{algorithmic}
\end{algorithm}

\KHH{
\begin{thm}
\label{thm:mr_mc_2d}
For multicast on a two-dimensional network, suppose \KH{each station uses}{all stations use} a sufficiently fine resolution (\textit{e.g.}, the upper bound in Theorem~\ref{thm:num_states_2d_mc}), so that the existence of collision-free configurations is guaranteed. 
Then, starting from an arbitrary initial configuration, Protocol~\ref{alg:mr_mc} will \khh{converge to}{, after a sufficiently long time, \KH{result in}{lead to}} a collision-free configuration. 
\end{thm}
}{}

Next we discuss how station $\mathbf{r}$ \KHH{exchanges}{obtain} the state information \KHH{required}{it requires} in Protocol~\ref{alg:mr_mc}. 
\KHH{A naive scheme is to let stations code the state information of different links into separate messages. 
However, links having the same transmitter share some common state information. 
Therefore, we introduce the following two-step message exchange which exploits this redunduncy to reduce the amount of information exchange between stations. }{}

The first step is to compute the state information of link $a$'s one-hop link-neighborhood \KHH{for all $a$ such that}{where} $\mathbf{r}$ is the transmitter of link $a$, \textit{i.e.}, $a=(\mathbf{r},\mathbf{r}^\prime)$. 
In the $t$-th cycle, station $\mathbf{r}$ collects $\bigl\lbrace X_{\mathbf{r}^\prime}(t)\bigr\rbrace_{\mathbf{r}^\prime\in V_\mathbf{r}}$. 
Then $\mathbf{r}$ constructs the following \textit{disjoint} sets: 
\begin{enumerate}
\item
common state information: $\mathcal{C}_\mathbf{r}(t)=\lbrace X_{\mathbf{r}^{\prime\prime}}(t)\rbrace_{\mathbf{r}^{\prime\prime}\colon\mathbf{r}\in D_{\mathbf{r}^{\prime\prime}}}$ (this includes the states of one-hop peers of $\mathbf{r}$ having $\mathbf{r}$ as an intended receiver, notice this state information is \textit{common to all links having $\mathbf{r}$ as the transmitter}); 
\item
self state information: $\mathcal{S}_\mathbf{r}(t)=\lbrace X_\mathbf{r}(t)\rbrace$ (this includes $\mathbf{r}$'s state, notice this state information is \textit{common only to all links in $M_\mathbf{r}$}); 
\item
link-specific state information for $(\mathbf{r},\mathbf{r}^\prime)$ where $\mathbf{r}^\prime\in V_\mathbf{r}$: $\mathcal{L}_{\mathbf{r},\mathbf{r}^\prime}(t)=\lbrace X_{\mathbf{r}^\prime}(t)\rbrace$ if $\mathbf{r}\notin D_{\mathbf{r}^\prime}$, and $\mathcal{L}_{\mathbf{r},\mathbf{r}^\prime}(t)=\emptyset$ otherwise (this includes $\mathbf{r}^\prime$'s state if $\mathbf{r}^\prime$ transmits to stations other than $\mathbf{r}$. 
\end{enumerate}
For $a=(\mathbf{r},\mathbf{r}^\prime)$, if $a\in M_\mathbf{r}$, the union $\mathcal{C}_\mathbf{r}(t)\cup\mathcal{S}_\mathbf{r}(t)\cup\mathcal{L}_{\mathbf{r},\mathbf{r}^\prime}(t)$ is the state information in the one-hop link-neighborhood of link $a$; otherwise if $a\notin M_\mathbf{r}$, the corresponding state information is the union $\mathcal{C}_\mathbf{r}(t)\cup\mathcal{L}_{\mathbf{r},\mathbf{r}^\prime}(t)$. 
\KHH{As an example, consider the network shown in Fig.~\ref{fig:onehopstateinfo_network}. 
A solid arrow is a link between a transmitter and an intended receiver, while a dashed arrow is a link between a transmitter and a nonintended receiver.
Fig.~\ref{fig:onehopstateinfo_table} shows how station $\mathbf{r}_0$ computes the state information of links $a_1,a_2,a_3,a_4$ following the above procedure. 
\K{Since $\mathbf{r}_0$ is an intended receiver for both $\mathbf{r}_2$ and $\mathbf{r}_3$, the states of both $\mathbf{r}_2$ and $\mathbf{r}_3$ are common to all links $a_1,a_2,a_3,a_4$, hence $\mathcal{C}_{\mathbf{r}_0}(t)=\lbrace X_{\mathbf{r}_2},X_{\mathbf{r}_3}\rbrace$.
Because $\mathbf{r}_0$ transmits to $\mathbf{r}_1$ and $\mathbf{r}_2$ only, $\mathcal{S}_{\mathbf{r}_0}(t)=\lbrace X_{\mathbf{r}_0}\rbrace$ is the state information common only to links $a_1,a_2$. 
Since $\mathbf{r}_1$ and $\mathbf{r}_4$ transmit to stations other than $\mathbf{r}_0$, $X_{\mathbf{r}_1}$ and $X_{\mathbf{r}_4}$ are included in $\mathcal{L}_{\mathbf{r}_0,\mathbf{r}_1}(t)$ and $\mathcal{L}_{\mathbf{r}_0,\mathbf{r}_4}(t)$\khh{,}{} respectively. 
}{}
}{}

\KHH{}{
\begin{table}[t]
\centering
\caption{One-hop state information of link $a$ with station $\mathbf{r}_0$ being the transmitter.}
\label{table:onehop_state}
\begin{tabular}{||c||c|c|c||}\hline
$a=(\mathbf{r},\mathbf{r}^\prime)$&$\mathcal{C}_\mathbf{r}(t)$&$\mathcal{S}_\mathbf{r}(t)$&$\mathcal{L}_{\mathbf{r},\mathbf{r^\prime}}(t)$\\\hline
$a_1$&$X_{\mathbf{r}_2},X_{\mathbf{r}_3}$&$X_{\mathbf{r}_5}$&$X_{\mathbf{r}_1}$\\\hline
$a_2$&$X_{\mathbf{r}_2},X_{\mathbf{r}_3}$&$X_{\mathbf{r}_5}$&empty\\\hline
$a_3$&$X_{\mathbf{r}_2},X_{\mathbf{r}_3}$&not included&empty\\\hline
$a_4$&$X_{\mathbf{r}_2},X_{\mathbf{r}_3}$&not included&$X_{\mathbf{r}_4}$\\\hline
\end{tabular}
\end{table}}

\begin{figure}[t]
\centering
\subfigure[]{
\includegraphics[width=2.7in]{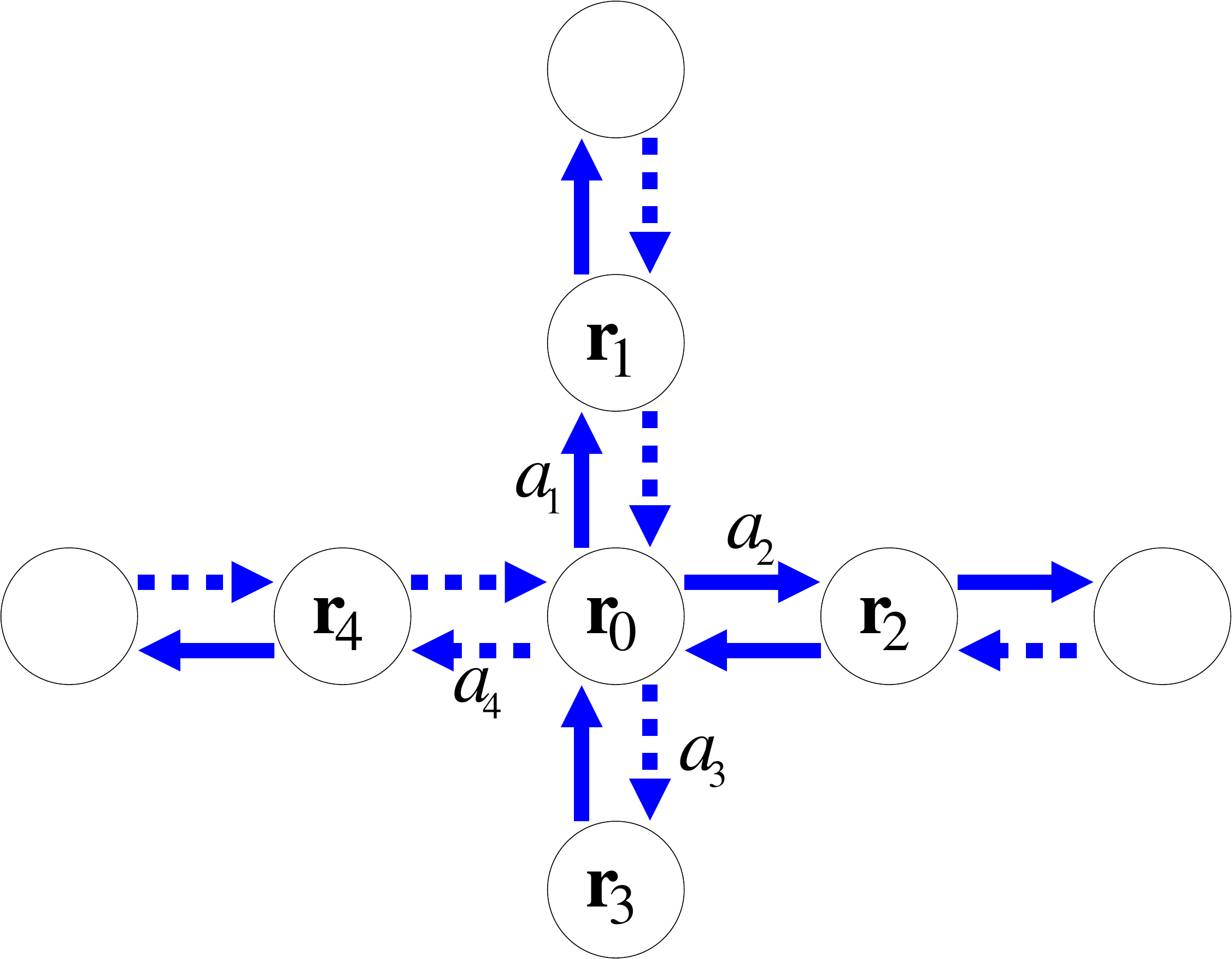}
\label{fig:onehopstateinfo_network}
}
\subfigure[]{
\includegraphics[width=3.5in]{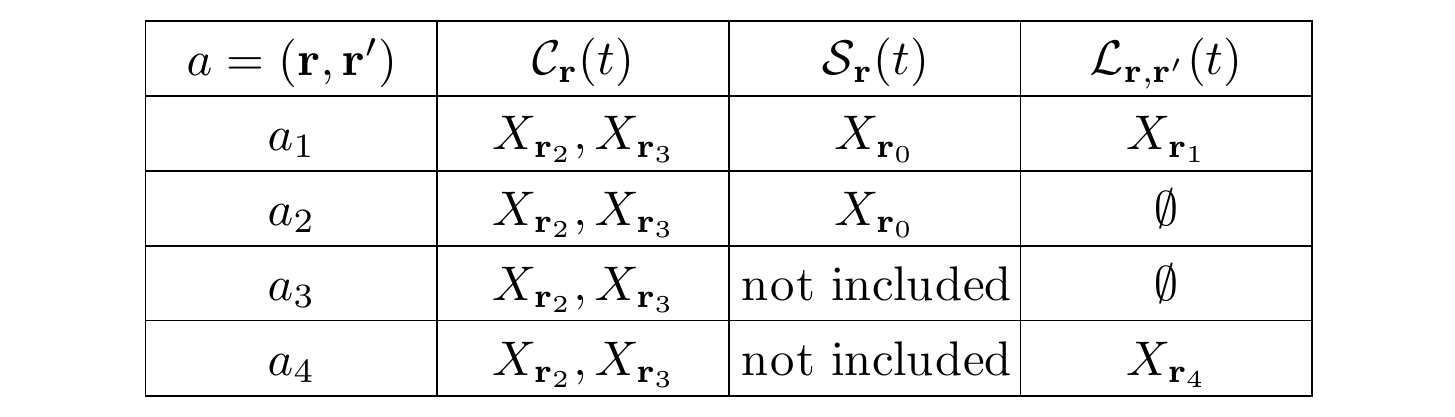}
\label{fig:onehopstateinfo_table}
}
\caption[]{One-hop state information of links $a_1,a_2,a_3,a_4$ with station $\mathbf{r}_0$ being the transmitter.}
\label{fig:onehopstateinfo}
\end{figure}

The \KHH{second}{next} step is to let $\mathbf{r}$ broadcast $\mathcal{C}_\mathbf{r}(t)$, $\mathcal{S}_\mathbf{r}(t)$, and $\mathcal{L}_{\mathbf{r},\mathbf{r}^\prime}(t)$ for $\mathbf{r}^\prime\in V_\mathbf{r}$. 
Note that any link in $\cup_{a\in M_\mathbf{r}}(L_a\cup\lbrace a\rbrace)$, \textit{i.e.}, the set of links\K{}{ in Protocol~\ref{alg:mr_mc}} that cast votes on $\mathbf{r}$'s next state, must be one of the followings: 
\begin{enumerate}
\item
$(\mathbf{r},\mathbf{r}^\prime)$ where $\mathbf{r}^\prime\in D_\mathbf{r}$, \textit{i.e.}, a link used by station $\mathbf{r}$,
\item
$(\mathbf{r}^\prime,\mathbf{r})$ where $\mathbf{r}^\prime\in V_\mathbf{r}\setminus D_\mathbf{r}$, \textit{i.e.}, the link from a nonintended receiver of station $\mathbf{r}$ to station $\mathbf{r}$,
\item
$(\mathbf{r}^\prime,\mathbf{r}^{\prime\prime})$ where $\mathbf{r}^\prime\in D_\mathbf{r}$ and $\mathbf{r}^{\prime\prime}\in V_{\mathbf{r}^\prime}$, \textit{i.e.}, a link originated from an intended receiver of station $\mathbf{r}$. 
\end{enumerate}
The transmitters of all these links are within the one-hop neighborhood of $\mathbf{r}$. 
\KHH{After station $\mathbf{r}$ receives the state information from its one-hop peer $\mathbf{r}^\prime$ in the second step, if $\mathbf{r}^\prime$ is an intended receiver of $\mathbf{r}$, then $\mathbf{r}$ needs to recover the state information of \textit{all} links with $\mathbf{r}^\prime$ as the transmitter; otherwise, $\mathbf{r}$ only needs to recover the state information of link $(\mathbf{r}^\prime,\mathbf{r})$. 
Here, it is assumed that station $\mathbf{r}$ knows $D_{\mathbf{r}^\prime}$ for all $\mathbf{r}^\prime\in V_\mathbf{r}$ so that recovery of state information of all links is possible; this can be done by letting each station broadcast a list of its intended receivers while setting up a multicast session.}{}
Therefore, station $\mathbf{r}$ can construct the one-hop state information of any link $a^\prime\in\cup_{a\in M_\mathbf{r}}(L_a\cup\lbrace a\rbrace)$, and then select\khh{}{s} its state at the $(t+1)$-st cycle following Protocol~\ref{alg:mr_mc}. 

\KHH{The amount of information exchange can be characterized as follows:
\begin{enumerate}
\item
In the first step, station $\mathbf{r}$ broadcasts its own state, which consists of $l_\mathbf{r}$ bits.
Station $\mathbf{r}$ \khh{also}{} broadcasts\khh{}{ also} its identity, which helps its one-hop peers partition the collected state information into the disjoint sets described above. 
\item
In the second step, station $\mathbf{r}$ broadcasts the states of itself and all its one-hop peers (which is already partitioned as described above), which consist of $l_\mathbf{r}+\sum_{\mathbf{r}^\prime\in V_\mathbf{r}}l_{\mathbf{r}^\prime}$ bits.
Station $\mathbf{r}$ also broadcasts its identity here, to help its one-hop peers recover the state information of each link.
\end{enumerate}
}{In the above process, station $\mathbf{r}$ needs to broadcast $2l_\mathbf{r}+\sum_{\mathbf{r}^\prime\in N_\mathbf{r}}l_{\mathbf{r}^\prime}$ bits for the state information, and in both steps, station $\mathbf{r}$ needs to broadcast its identity together with the state information.}

\KHH{}{
By considering again the analogy between multicast on $G$ and broadcast on $G^L$, we have the following convergence result. 

\begin{thm}
\label{thm:mr_mc_2d}
For multicast on a two-dimensional network, suppose \KH{each station uses}{all stations use} a sufficiently fine resolution (\textit{e.g.}, the upper bound in Theorem~\ref{thm:num_states_2d_mc}), so that the existence of collision-free configurations is guaranteed. 
Then, starting from an arbitrary initial configuration, Protocol~\ref{alg:mr_mc} will, after a sufficiently long time, \KH{result in}{lead to} a collision-free configuration. 
\end{thm}}

\begin{figure}[t]
\centering
\subfigure[Convergence time]{
\includegraphics[width=3.5in]{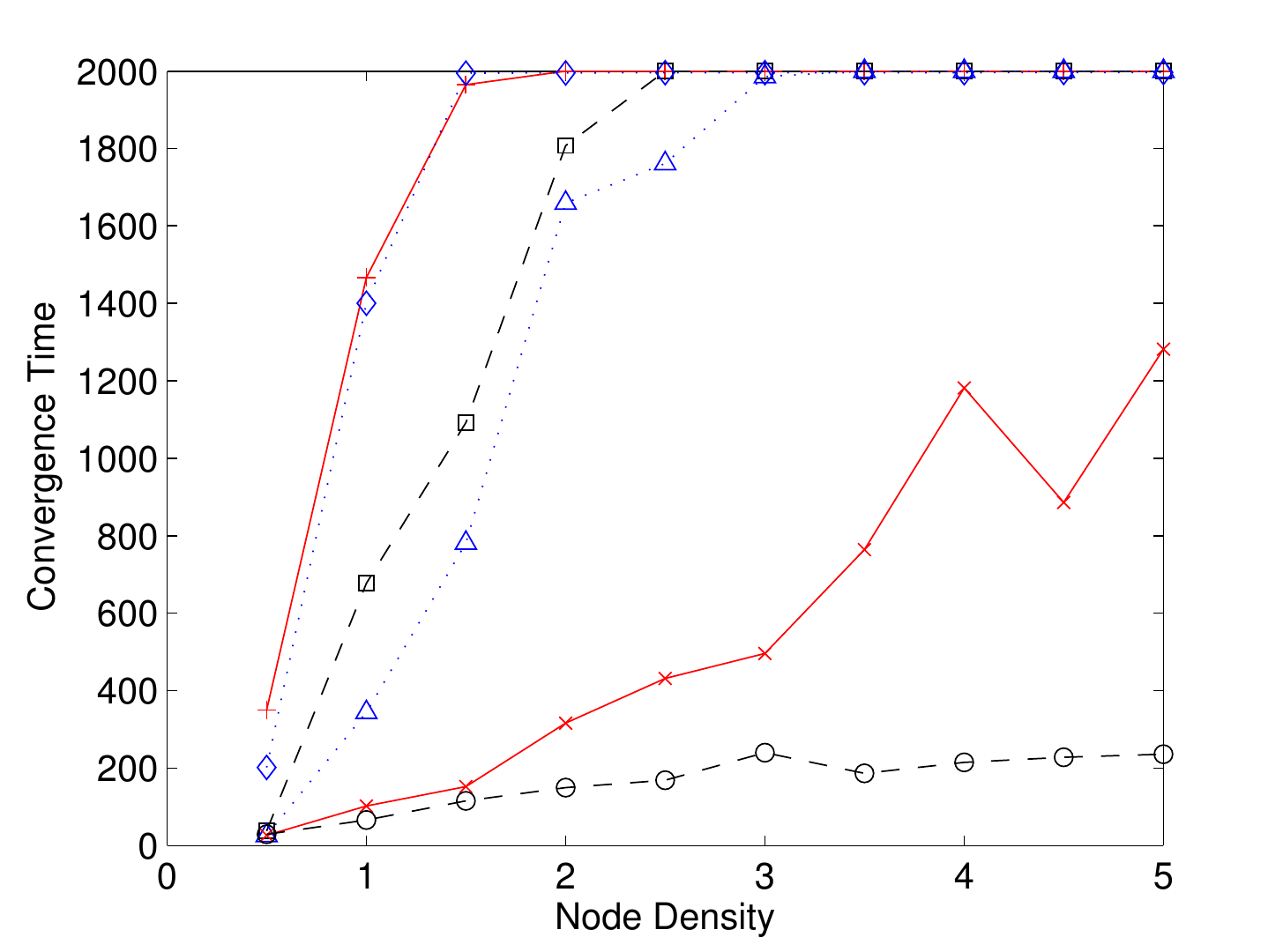}
\label{fig:time_mc_2d_20}
}
\hspace{-0.4in}
\subfigure[Convergence percentage]{
\includegraphics[width=3.5in]{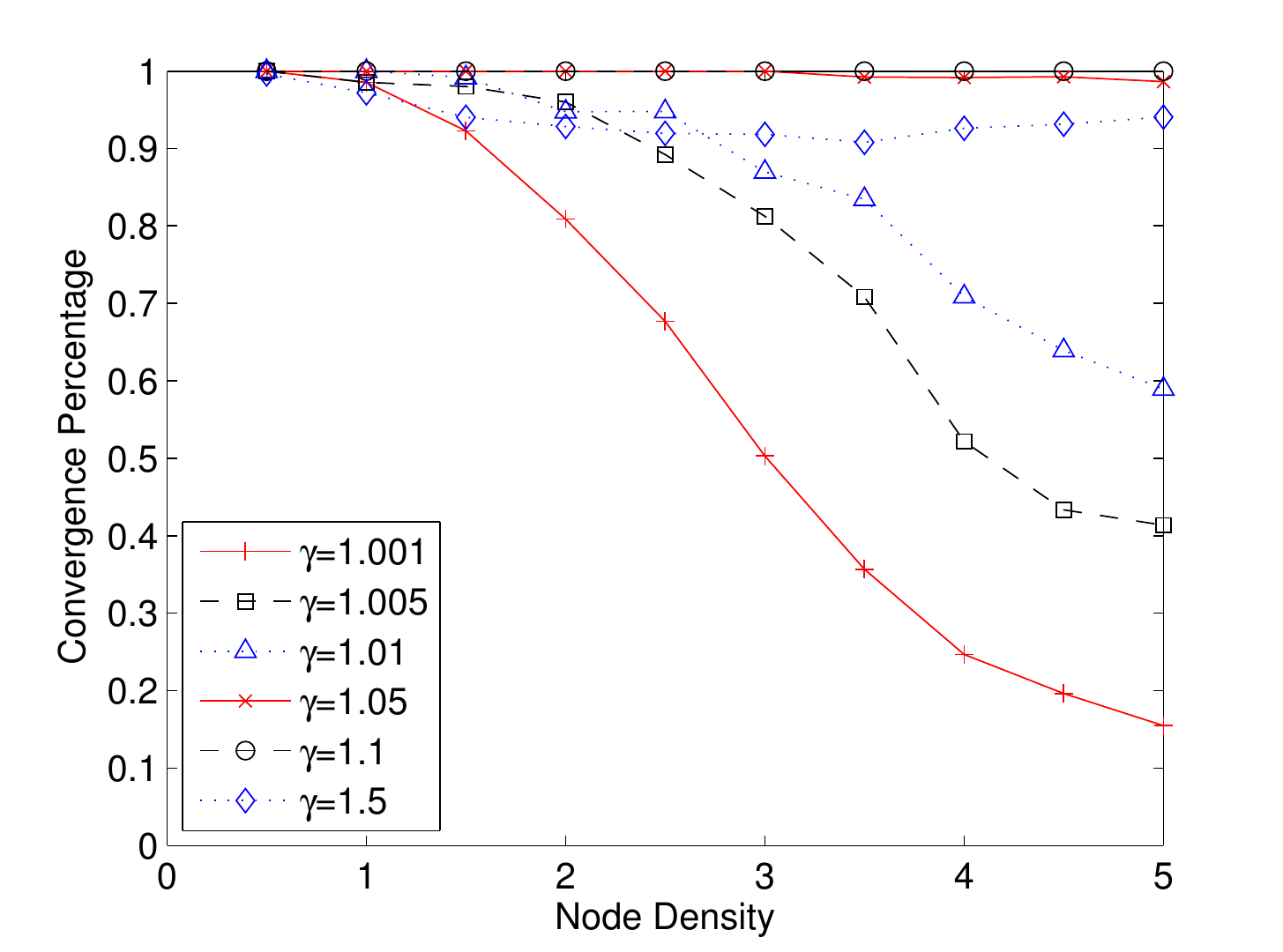}
\label{fig:percentage_mc_2d_20}
}
\caption[]{Simulations of the multi-resolution protocol with annealing for multicast, $q=0.2$ (the same legend applies to both figures).}
\label{fig:annealing_mc_2d_20}
\end{figure}

\begin{figure}[t]
\centering
\subfigure[Convergence time]{
\includegraphics[width=3.5in]{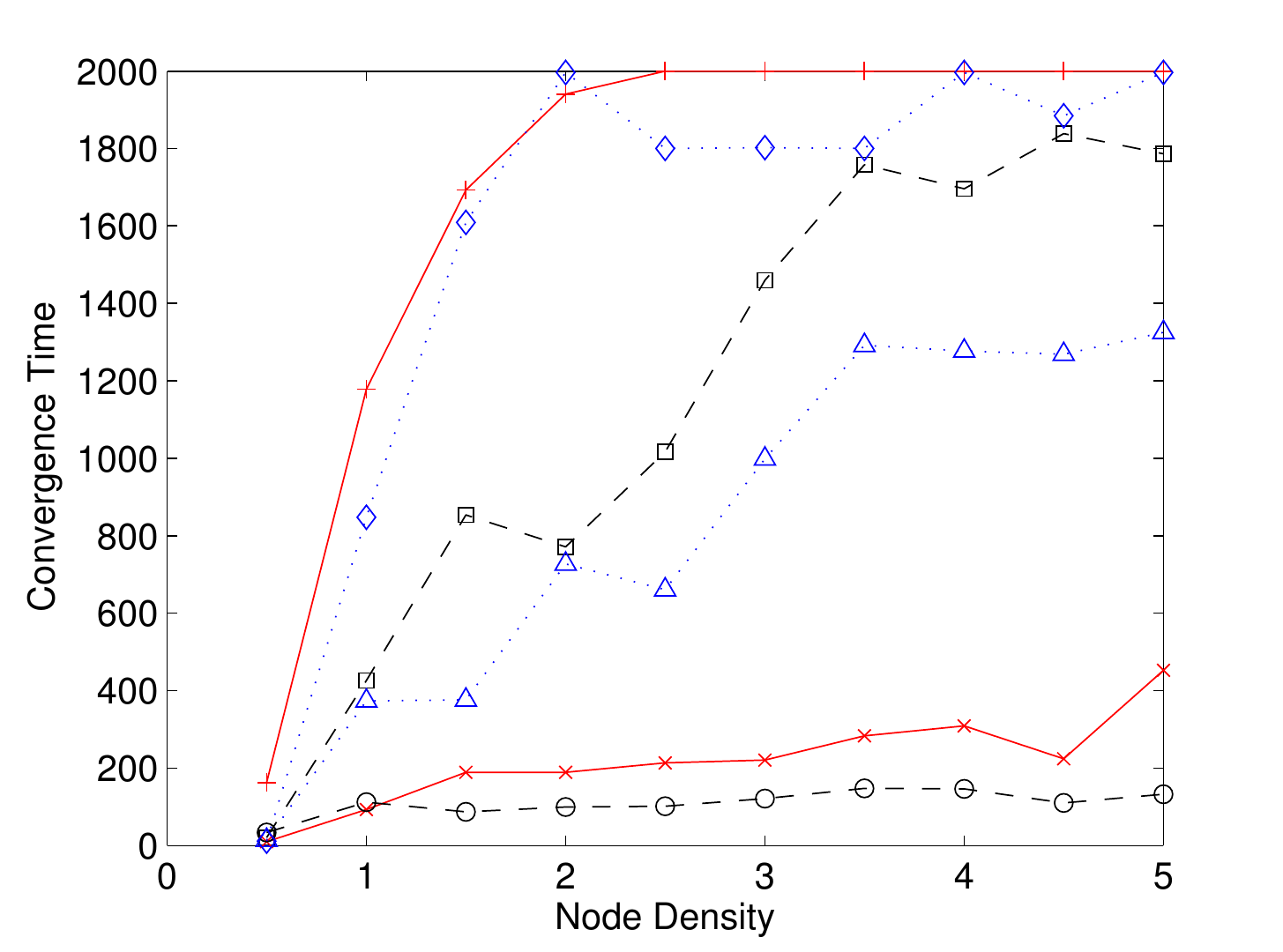}
\label{fig:time_mc_2d_80}
}
\hspace{-0.4in}
\subfigure[Convergence percentage]{
\includegraphics[width=3.5in]{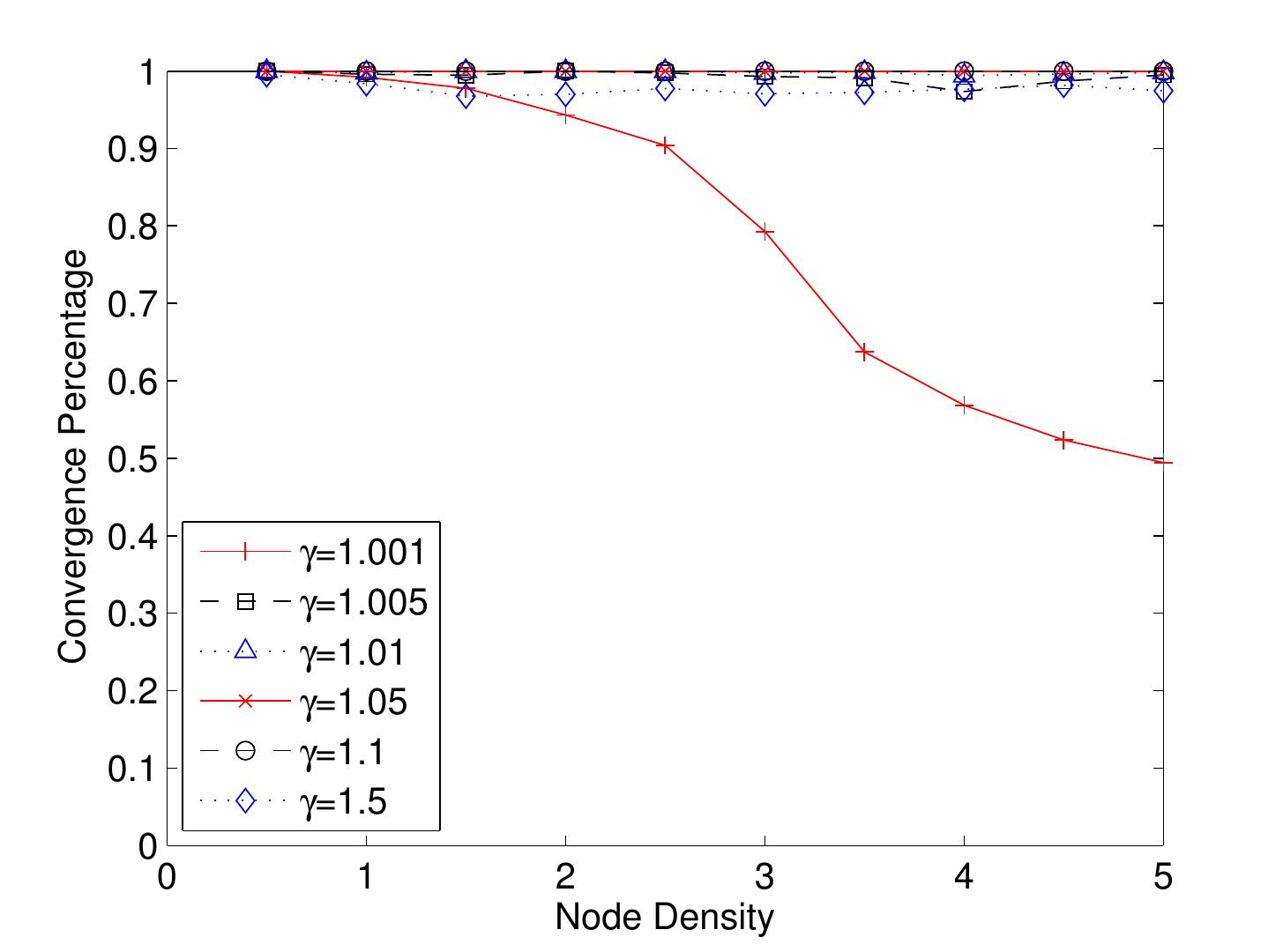}
\label{fig:percentage_mc_2d_80}
}
\caption[]{Simulations of the multi-resolution protocol with annealing for multicast, $q=0.8$ (the same legend applies to both figures).}
\label{fig:annealing_mc_2d_80}
\end{figure}

Figs.~\ref{fig:annealing_mc_2d_20} and~\ref{fig:annealing_mc_2d_80} show \khh{a}{the} simulation of Protocol~\ref{alg:mr_mc} \khh{for a}{on} two-dimensional network\khh{}{s}. 
In the simulation, any one-hop peer of station $\mathbf{r}$ is an intended receiver of the multicast by station $\mathbf{r}$ with probability $q$, independent of other one-hop peers. 
All other simulation settings are the same as those for broadcast. 
Each station uses the \KH{lower bound on the resolution predicted by Theorem~\ref{thm:num_states_2d_mc}}{resolution predicted by Theorem~\ref{thm:num_states_2d_mc_lb}} and executes Protocol~\ref{alg:mr_mc}. 
Each station refines its resolutions if necessary, until the local configuration is collision-free, or the upper bound given by Theorem~\ref{thm:num_states_2d_mc} is reached, whichever occurs first. 
The simulation results are similar to those for broadcast in Section~\ref{sec:2d}. 
It appears that when $q$ is larger, the convergence time is shorter and the convergence percentage is higher. 
This suggests that as $q$ is larger, the lower bound given by Theorem~\ref{thm:num_states_2d_mc} is accurate enough so fewer stations need to refine their resolutions, which helps speed up the convergence. 

\section{Multi-Channel Networks}
\label{sec:multich}

In this section we assume there are $K$ orthogonal channels in the network. 
A station can either transmit on one channel only, or listen to all channels simultaneously at any time, \textit{i.e.}, stations are half-duplex.\footnote{This is just one of several possibilities. For half-duplex constraints with multiple channels, similar ideas can apply to other models, \textit{e.g.}, if stations can listen to only one channel at a time.}
\KHH{In this case, the state of a station represents \textit{both the slot $s$ and the channel $\omega$} over which the station transmits, \textit{i.e.}, $X_\mathbf{r}(t)=(\omega,s)$. }{}

To estimate the lower and upper bounds on the resolutions for broadcast\KHH{}{ and multicast} with $K$ orthogonal channels, notice that the best possible scenario in station $\mathbf{r}$'s one-hop neighborhood is that $\mathbf{r}$ occupies one slot to transmit, and in the remaining slots, $\mathbf{r}$ receives one packet on each channel; while the worst situation in station $\mathbf{r}$'s two-hop neighborhood is that every station within this two-hop neighborhood must transmit at different times. 
Hence we have the following results. 

\KH{
\begin{thm}
\label{thm:num_states_2d_multich}
\khh{A lower bound on the needed resolution for a collision-free configuration for broadcast with $K$ orthogonal channels to exist is given by}{\kh{There exists a self-stabilizing protocol for broadcast with $K$ orthogonal channels for which the resolution of each station $\mathbf{r}$ is lower bounded by}{For broadcast with $K$ orthogonal channels, \kh{to allow stations to have a fair share of the wireless channel with their respective neighbors while avoiding collisions, a}{\K{a}{the}} lower bound\K{}{ $\underline{l}_\mathbf{r}$} on the resolution used by station $\mathbf{r}$ is}}\K{
\begin{equation}
\label{eqn:resolution_multich_lb}
\underline{l}_\mathbf{r}=\biggl\lceil\log_2\biggl(\max_{\mathbf{r}^\prime\in V_\mathbf{r}\cup\lbrace\mathbf{r}\rbrace}\underline{w}_{\mathbf{r}^\prime}\biggr)\biggr\rceil,
\end{equation}
where $\underline{w}_\mathbf{r}=1+\frac{\lvert V_\mathbf{r}\rvert}{K}$\khh{.}{;}
}{ computed as follows:
\begin{enumerate}
\item
$\underline{w}_\mathbf{r}=1+\frac{\lvert V_\mathbf{r}\rvert}{K}$,
\item
$\underline{l}_\mathbf{r}=\lceil\log_2(\max_{\mathbf{r}^\prime\in V_\mathbf{r}\cup\lbrace\mathbf{r}\rbrace}\underline{w}_{\mathbf{r}^\prime})\rceil$.
\end{enumerate}}\kh{\khh{A corresponding upper bound is given by}{and upper bounded by}}{\K{A}{The} \KHH{corresponding}{} upper bound\K{}{ $\overline{l}_\mathbf{r}$}\KHH{}{ on the resolution used by station $\mathbf{r}$} is}\K{
\begin{equation}
\label{eqn:resolution_multich_ub}
\overline{l}_\mathbf{r}=\biggl\lceil\log_2\biggl(\max_{\mathbf{r}^\prime\in V_\mathbf{r}^2\cup\lbrace\mathbf{r}\rbrace}\overline{w}_{\mathbf{r}^\prime}\biggr)\biggr\rceil,
\end{equation}
where $\overline{w}_\mathbf{r}=1+\lvert V_\mathbf{r}^2\rvert$.
}{ computed as follows:
\begin{enumerate}
\item
$\overline{w}_\mathbf{r}=1+\lvert V_\mathbf{r}^2\rvert$,
\item
$\overline{l}_\mathbf{r}=\lceil\log_2(\max_{\mathbf{r}^\prime\in V_\mathbf{r}^2\cup\lbrace\mathbf{r}\rbrace}\overline{w}_{\mathbf{r}^\prime})\rceil$.
\end{enumerate}}
\end{thm}
}{
\begin{thm}
\label{thm:num_states_2d_multich_lb}
For broadcast with $K$ orthogonal channels, the lower bound $\underline{l}_\mathbf{r}$ on the resolution used by station $\mathbf{r}$ is computed as follows:
\begin{enumerate}
\item
$\underline{w}_\mathbf{r}=1+\frac{\lvert N_\mathbf{r}\rvert}{K}$,
\item
 $\underline{l}_\mathbf{r}=\lceil\log_2(\max_{\mathbf{r}^\prime\in N_\mathbf{r}\cup\lbrace\mathbf{r}\rbrace}\underline{w}_{\mathbf{r}^\prime})\rceil$.
\end{enumerate}
\end{thm}

\begin{thm}
\label{thm:num_states_2d_multich_ub}
For broadcast with $K$ orthogonal channels, the upper bound $\overline{l}_\mathbf{r}$ on the resolution used by station $\mathbf{r}$ is computed as follows:
\begin{enumerate}
\item
$\overline{w}_\mathbf{r}=1+\lvert N_\mathbf{r}^2\rvert$,
\item
$\overline{l}_\mathbf{r}=\lceil\log_2(\max_{\mathbf{r}^\prime\in N_\mathbf{r}^2\cup\lbrace\mathbf{r}\rbrace}\overline{w}_{\mathbf{r}^\prime})\rceil$.
\end{enumerate}
Here, $N_\mathbf{r}^2$ is the set of one-hop or two-hop peers of $\mathbf{r}$. 
\end{thm}
}

Similar arguments provide the corresponding lower and upper bounds on the resolution for multicast. 

\KH{
\begin{thm}
\label{thm:num_states_2d_mc_multich}
\khh{A lower bound on the needed resolution for a collision-free configuration for multicast with $K$ orthogonal channels to exist}{\kh{There exists a self-stabilizing protocol for multicast with $K$ orthogonal channels for which the lower bound on the resolution of each station $\mathbf{r}$}{For multicast with $K$ orthogonal channels, \kh{to allow stations to have a fair share of the wireless channel with their respective neighbors while avoiding collisions, a}{\K{a}{the}} lower bound\K{}{ $\underline{l}_\mathbf{r}$} on the resolution used by station $\mathbf{r}$}} is computed as follows:
\K{
\begin{IEEEeqnarray}{rCl}
\underline{I}_a&=&\lbrace\mathbf{r}^\prime\colon M_{\mathbf{r}^\prime}\cap(L_a\cup\lbrace a\rbrace)\ne\emptyset\rbrace,\\
\underline{w}_a&=&1+\frac{\lvert\underline{I}_a\rvert-1}{K},\\
\underline{l}_a&=&\max_{a^\prime\in L_a\cup\lbrace a\rbrace}\underline{w}_{a^\prime},\\
\underline{l}_\mathbf{r}&=&\biggl\lceil\log_2\biggl(\max_{a\in M_\mathbf{r}}\underline{l}_a\biggr)\biggr\rceil.
\end{IEEEeqnarray}
}{
\begin{enumerate}
\item
$\underline{I}_a=\lbrace\mathbf{r}^\prime\colon M_{\mathbf{r}^\prime}\cap(L_a\cup\lbrace a\rbrace)\ne\emptyset\rbrace$,
\item
$\underline{w}_a=1+\frac{\lvert\underline{I}_a\rvert-1}{K}$,
\item
$\underline{l}_a=\max_{a^\prime\in L_a\cup\lbrace a\rbrace}\underline{w}_{a^\prime}$,
\item
$\underline{l}_\mathbf{r}=\lceil\log_2(\max_{a\in M_\mathbf{r}}\underline{l}_a)\rceil$.
\end{enumerate}}\K{A}{The} \KHH{corresponding}{} upper bound\K{}{ $\overline{l}_\mathbf{r}$}\KHH{}{ on the resolution used by station $\mathbf{r}$} is computed as follows:
\K{
\begin{IEEEeqnarray}{rCl}
\overline{I}_\mathbf{r}&=&\lbrace\mathbf{r}^\prime\colon M_{\mathbf{r}^\prime}\cap\cup_{a\in M_\mathbf{r}}(L_a^2\cup\lbrace a\rbrace)\ne\emptyset\rbrace,\\
\overline{w}_\mathbf{r}&=&\lvert\overline{I}_\mathbf{r}\rvert,\\
\overline{l}_\mathbf{r}&=&\biggl\lceil\log_2\biggl(\max_{\mathbf{r}^\prime\in\overline{I}_\mathbf{r}}\overline{w}_{\mathbf{r}^\prime}\biggr)\biggr\rceil.
\end{IEEEeqnarray}
}{
\begin{enumerate}
\item
$\overline{I}_\mathbf{r}=\lbrace\mathbf{r}^\prime\colon M_{\mathbf{r}^\prime}\cap\cup_{a\in M_\mathbf{r}}(L_a^2\cup\lbrace a\rbrace)\ne\emptyset\rbrace$,
\item
$\overline{w}_\mathbf{r}=\lvert\overline{I}_\mathbf{r}\rvert$,
\item
$\overline{l}_\mathbf{r}=\lceil\log_2(\max_{\mathbf{r}^\prime\in\overline{I}_\mathbf{r}}\overline{w}_{\mathbf{r}^\prime})\rceil$.
\end{enumerate}}
\end{thm}
}{
\begin{thm}
\label{thm:num_states_2d_mc_multich_lb}
For multicast with $K$ orthogonal channels, the lower bound $\underline{l}_\mathbf{r}$ on the resolution used by station $\mathbf{r}$ is computed as follows:
\begin{enumerate}
\item
$\underline{I}_a=\lbrace\mathbf{r}^\prime\colon S_{\mathbf{r}^\prime}\cap(L_a\cup\lbrace a\rbrace)\ne\emptyset\rbrace$,
\item
$\underline{w}_a=1+\frac{\lvert\underline{I}_a\rvert-1}{K}$,
\item
$\underline{l}_a=\max_{a^\prime\in L_a\cup\lbrace a\rbrace}\underline{w}_{a^\prime}$,
\item
$\underline{l}_\mathbf{r}=\lceil\log_2(\max_{a\in S_\mathbf{r}}\underline{l}_a)\rceil$.
\end{enumerate}
Here, $L_a$ is the set of one-hop link-peers of link $a$. 
\end{thm}

\begin{thm}
\label{thm:num_states_2d_mc_multich_ub}
For multicast with $K$ orthogonal channels, the upper bound $\overline{l}_\mathbf{r}$ on the resolution used by station $\mathbf{r}$ is computed as follows:
\begin{enumerate}
\item
$\overline{I}_\mathbf{r}=\lbrace\mathbf{r}^\prime\colon S_{\mathbf{r}^\prime}\cap\cup_{a\in S_\mathbf{r}}(L_a^2\cup\lbrace a\rbrace)\ne\emptyset\rbrace$,
\item
$\overline{w}_\mathbf{r}=\lvert\overline{I}_\mathbf{r}\rvert$,
\item
$\overline{l}_\mathbf{r}=\lceil\log_2(\max_{\mathbf{r}^\prime\in\overline{I}_\mathbf{r}}\overline{w}_{\mathbf{r}^\prime})\rceil$.
\end{enumerate}
Here, $L_a^2$ is the set of one-hop or two-hop link-peers of link $a$. 
\end{thm}
}

\KHH{
When there are multiple channels available, stations need to let their peers know which slot and channel they use to transmit. 
If a dedicated control channel (which is orthogonal to the $K$ channels for data transmission) is used to exchange state information, then $l_\mathbf{r}+\lceil\log_2K\rceil$ bits are required to represent station $\mathbf{r}$'s state: $l_\mathbf{r}$ bits for the slot, and $\lceil\log_2K\rceil$ bits for the channel. 
Alternatively, if there are control frames preceding each cycle, and these can be used to exchange state information on the $K$ channels for data transmission, a station can save the extra $\lceil\log_2K\rceil$ bits as follows: it broadcasts the state information on the $\omega$-th channel if the state information indicates a transmission on the $\omega$-th channel. 

The multi-resolution protocol for broadcast on networks with multiple channels is shown as Protocol~\ref{alg:mr_bc_multich}. 
The main difference from Protocol~\ref{alg:mr_bc} is that before station $\mathbf{r}$ computes the votes using the state information from station $\mathbf{r}^\prime$'s one-hop neighborhood, station $\mathbf{r}$ assumes the states $(\omega,s)$ for all $\omega$ are occupied by station $\mathbf{r}^\prime$, where $s$ is the slot currently occupied by station $\mathbf{r}^\prime$ (line $3$ in Protocol~\ref{alg:mr_bc_multich}).
This is due to the half-duplex constraint: if station $\mathbf{r}^\prime$ transmits in slot $s$, then it cannot receive on \textit{any} channel in slot $s$, meaning that packets transmitted by any one-hop peer in this slot experience collisions. 
The multi-resolution protocol for multicast on networks with multiple channels can be constructed similarly. 

\begin{algorithm}[t]
\caption{Multi-Resolution MAC Protocol for Broadcast on Networks with Multiple Channels}
\label{alg:mr_bc_multich}
\begin{algorithmic}[1]
\WHILE{\khh{station $\mathbf{r}$ is active}{}}
\STATE
$\mathbf{r}$ sets the votes on all states to zero. 
\FOR{$\mathbf{r}^\prime\in V_\mathbf{r}\cup\lbrace\mathbf{r}\rbrace$}
\STATE
Assume states $(\omega,s)$ for all $\omega$ are occupied by station $\mathbf{r}^\prime$, where $s$ is the slot currently occupied by $\mathbf{r}^\prime$.
\IF{$\mathbf{r}$ is the only station occupying its current state in station $\mathbf{r}^\prime$'s one-hop neighborhood} 
\STATE
\KH{$\mathbf{r}$'s current state is assigned a single vote of weight one}{A single vote of weight one on $\mathbf{r}$'s current state is given by $\mathbf{r}^\prime$}. 
\ELSE
\STATE
\KH{$\mathbf{r}$ determines which slots $s$ (according to $\mathbf{r}$'s resolution) are idle or have collisions in $\mathbf{r}^\prime$'s one-hop neighborhood}{$\mathbf{r}$ determines the states (according to $\mathbf{r}$'s resolution) that station $\mathbf{r}^\prime$ is idle or collides}. 
\STATE
A vote of weight $\frac{1}{Kn}$ is \khh{added}{given} to states $(\omega,s)$ for all $\omega$, where $n$ is the number of \KHH{slots $s$ determined above}{such slots}. 
\ENDIF
\ENDFOR
\IF{$n_{(\omega,s)}>0$ for multiple $(\omega,s)$'s, where $n_{(\omega,s)}$ is the total weight state $(\omega,s)$ receives}
\STATE
Replace $n_{(\omega,s)}$ by $n_{(\omega,s)}+\epsilon$, where $\epsilon>0$, for all $(\omega,s)$. 
\ENDIF
\STATE
$\mathbf{r}$ selects state $(\omega,s)$ with a probability proportional to $f(n_{(\omega,s)})$.
\ENDWHILE
\end{algorithmic}
\end{algorithm}
}{}

%\enlargethispage{-0.1in}

\section{Conclusion}
\label{sec:conclusion}

In this paper, we have proposed multi-resolution MAC protocols for wireless networks with arbitrary topolog\KH{ies}{y}. 
We have shown that \khh{collision-free schedules can be established}{collisions can be eliminated} in a distributed manner, by \khh{allowing stations to exchange limited state information}{letting stations make their transmission decisions based on recent transmission decisions of its peers only}. 
These protocols do not require all stations to use the same resolution, \textit{i.e.}, the same number of states or the same length of each slot\KH{}{, and zero collisions can still be achieved}. 

\KH{Future work should \K{investigate}{explore} the performance of the multi-resolution protocols under the signal-to-interference-\kh{and-}{}noise ratio (SINR) model. 
This model assumes a minimum SINR requirement at a receiver for successful reception and also takes into account cumulative interference from faraway transmissions, which is more realistic. 
Under this model, there may be a need to reconsider what messages should be exchanged among peers in order to eliminate collisions in a distributed manner. 
}{}

\section*{Acknowledgment}

The authors would like to thank Tianyi Li for his assistance in simulations. 

\bibliographystyle{IEEEtran}
\bibliography{IEEEabrv,bibliography}

\end{document}